\documentclass[12pt,oneside]{amsart}
\usepackage[utf8]{inputenc}
\usepackage[a4paper]{geometry}
\usepackage{mathtools}
\usepackage{amstext}
\usepackage{amsthm}
\usepackage{amssymb}
\usepackage{stackrel}
\usepackage{xargs}[2008/03/08]
\usepackage[unicode=true,pdfusetitle,
 bookmarks=true,bookmarksnumbered=false,bookmarksopen=false,
 breaklinks=false,pdfborder={0 0 1},backref=false,colorlinks=false]
 {hyperref}

\makeatletter

\DeclareTextSymbolDefault{\textquotedbl}{T1}

\numberwithin{equation}{section}
\numberwithin{figure}{section}
\theoremstyle{plain}
\newtheorem{thm}{\protect\theoremname}
\theoremstyle{definition}
\newtheorem{defn}[thm]{\protect\definitionname}
\theoremstyle{plain}
\newtheorem{lem}[thm]{\protect\lemmaname}
\theoremstyle{plain}
\newtheorem{cor}[thm]{\protect\corollaryname}
\theoremstyle{remark}
\newtheorem{rem}[thm]{\protect\remarkname}

\@ifundefined{date}{}{\date{}}
\usepackage{longtable}

\usepackage{tikz}
\usetikzlibrary{shapes}
\usetikzlibrary{arrows.meta}

\DeclareSymbolFont{extraup}{U}{zavm}{m}{n}
\DeclareMathSymbol{\varheart}{\mathalpha}{extraup}{86}
\DeclareMathSymbol{\vardiamond}{\mathalpha}{extraup}{87} 

\providecommand{\corollaryname}{Corollary}
\providecommand{\definitionname}{Definition}
\providecommand{\lemmaname}{Lemma}
\providecommand{\theoremname}{Theorem}

\makeatother

\providecommand{\corollaryname}{Corollary}
\providecommand{\definitionname}{Definition}
\providecommand{\lemmaname}{Lemma}
\providecommand{\remarkname}{Remark}
\providecommand{\theoremname}{Theorem}

\begin{document}
\global\long\def\C{\mathbb{C}}%
\global\long\def\Cd{\C^{\delta}}%
\global\long\def\Cprim{\C^{\delta,\circ}}%
\global\long\def\Cdual{\C^{\delta,\bullet}}%
\global\long\def\Od{\Omega^{\delta}}%
\global\long\def\Oprim{\Omega^{\delta,\circ}}%
\global\long\def\Odual{\Omega^{\delta,\bullet}}%
\global\long\def\en{\mathcal{\varepsilon}}%
\global\long\def\wone{\widehat{w}}%
\global\long\def\wtwo{w}%
\global\long\def\wo{\wone}%
\global\long\def\wt{\wtwo}%
\global\long\def\cZ{\mathcal{Z}}%

\global\long\def\tfixed{wired}%
\global\long\def\tfree{free}%
\global\long\def\tplus{plus}%
\global\long\def\tminus{minus}%
\global\long\def\sps{(p,q)}%
\global\long\def\fixed{\{\mathtt{\tfixed}\}}%
\global\long\def\free{\{\mathtt{\tfree}\}}%
\global\long\def\plus{\{\mathtt{\tplus}\} }%
\global\long\def\minus{\left\{  \mathtt{\tminus}\right\}  }%
\global\long\def\Cgr{\mathcal{C}}%
\global\long\def\sCgr{\mathit{c}}%
\global\long\def\Pf{\mathrm{Pf}}%
\global\long\def\vv{v_{1},\dots,v_{n}}%
\global\long\def\uu{u_{1},\dots,u_{m}}%
\global\long\def\svv{\sigma_{v_{1}}\ldots\sigma_{v_{n}}}%
\global\long\def\proj#1#2{\text{Proj}_{#1}(#2)}%

\global\long\def\muu{\mu_{v_{n+1}}\dots\mu_{v_{k}}}%
\global\long\def\pesm{\psi^{[\eta]}\!,\en,\mu,\sigma}%
\global\long\def\sfix#1{\sigma_{\mathrm{fix}}^{#1}}%
\global\long\def\any{\diamond}%
\global\long\def\anyother{\triangleright}%
\global\long\def\Opunc{\Omega^{\boxcircle}}%
\global\long\def\Onopunc{\Omega^{\boxempty}}%
\global\long\def\Oother{\Omega'}%
\global\long\def\const{\mathrm{const}}%
\global\long\def\P{\mathsf{\mathbb{P}}}%
\global\long\def\E{\mathsf{\mathbb{E}}}%
\global\long\def\sF{\mathcal{F}}%
\global\long\def\ind{\mathbb{I}}%
\global\long\def\R{\mathbb{R}}%
\global\long\def\Z{\mathbb{Z}}%
\global\long\def\N{\mathbb{N}}%
\global\long\def\Q{\mathbb{Q}}%
\global\long\def\C{\mathbb{C}}%
\global\long\def\Rsphere{\overline{\C}}%
\global\long\def\re{\Re\mathfrak{e}}%
\global\long\def\im{\Im\mathfrak{m}}%
\global\long\def\arg{\mathrm{arg}}%
\global\long\def\i{\mathfrak{i}}%
\global\long\def\eps{\varepsilon}%
\global\long\def\lamb{\lambda}%
\global\long\def\lambb{\bar{\lambda}}%
\global\long\def\D{\mathbb{D}}%
\global\long\def\HH{\mathbb{H}}%

\global\long\def\Deps{\Omega_{\eps}}%
\global\long\def\DD{\hat{\Omega}}%

\global\long\def\dist{\mathrm{dist}}%
\global\long\def\reg{\mathrm{reg}}%
\global\long\def\half{\frac{1}{2}}%
\global\long\def\sgn{\mathrm{sgn}}%
\global\long\def\bdry{\partial}%
\global\long\def\cl#1{\overline{#1}}%
\global\long\def\diam{\mathrm{diam}}%
\global\long\def\corr#1{\overline{#1}}%
\global\long\def\Corr#1#2{\E_{#1}(#2)}%
\global\long\def\corr#1#2{\langle#1\rangle_{#2}}%
\global\long\def\vipq#1#2#3{v_{#1;#2#3}}%
\global\long\def\pa{\partial}%
\global\long\def\tto#1{\stackrel{#1}{\longrightarrow}}%
\global\long\def\res{\text{res}}%
\global\long\def\u{u}%
\global\long\def\v{v}%
\global\long\def\z{z}%
\global\long\def\mod{\;\mathrm{mod\;}}%
\global\long\def\wind{\mathrm{wind}}%
\global\long\def\vz{z^{\bullet}}%
\global\long\def\fz{z^{\circ}}%
\global\long\def\CF{\mathfrak{C}}%
\global\long\def\RPF{\text{I}}%
\global\long\def\FFS#1#2#3{F_{#2}^{#3}(#1)}%
\global\long\def\ffs#1#2#3{f_{#2}^{#3}(#1)}%
\global\long\def\gfs#1#2#3{g_{#2}^{#3}(#1)}%
\global\long\def\drs#1{\eta_{#1}}%
\global\long\def\dbar{\overline{\partial}}%
\global\long\def\dual#1{\left(#1\right)^{*}}%
\global\long\def\Eod{\Od_{+}}%
\global\long\def\GammaR{\Gamma_{\R}}%
\global\long\def\GammaiR{\Gamma_{i\R}}%
\global\long\def\ccor#1{\left\langle #1\right\rangle }%
\global\long\def\bar#1{\overline{#1}}%
\global\long\def\anypsi{\psi^{*}}%
\global\long\def\Op{\mathcal{O}}%
\global\long\def\crossing{c}%
\global\long\def\Fdual{F^{\mathrm{dual}}}%
\global\long\def\zz{z_{1},z_{2}}%
\newcommandx\norm[1][usedefault, addprefix=\global, 1=]{n_{#1}}%
\global\long\def\Ocvrc{\Omega_{\cvr}}%
\global\long\def\crad{\text{crad}}%
\global\long\def\feta{f^{[\eta]}}%

\global\long\def\reg{\sharp}%
\global\long\def\regg{*}%
\global\long\def\coefA{\mathcal{A}}%
\global\long\def\pbar#1{#1^{\star}}%
\global\long\def\fdag{f^{\star}}%
\global\long\def\CorrO#1{\langle#1\rangle_{\Omega}}%
\global\long\def\formL{\mathcal{L}}%
\global\long\def\appe{\approx_{\eps}}%
\global\long\def\jayhat{\hat{j}}%
\global\long\def\sqr#1{(#1)^{\frac{1}{2}}}%
\global\long\def\M{\mathbb{\mathcal{M}}}%
\global\long\def\Mtilde{\mathcal{M^{\star}}}%
\global\long\def\Meps{\mathbb{\mathcal{M}}_{\eps}}%
\global\long\def\Mepso{\mathbb{\mathcal{M}}_{\epsilon_{1}}}%
\global\long\def\Mhat{\mathbb{\mathcal{M}}_{0}}%
\global\long\def\dZ{\mathbb{C}^{\delta}}%
\global\long\def\dZd{\dZ^{\delta}}%
\global\long\def\dZprim{\mathbb{C}^{\delta,\circ}}%
\global\long\def\dZdual{\mathbb{C}^{\delta,\star}}%

\global\long\def\T{\mathbb{T}}%
\global\long\def\Ttilde{\widetilde{\mathbb{T}}}%
\global\long\def\Td{\mathbb{T}^{\delta}}%
\global\long\def\Tdouble{\hat{\mathbb{T}}^{\delta}}%
\global\long\def\Lambdad{\Lambda^{\delta}}%
\global\long\def\Tprim{\T^{\delta,\circ}}%
\global\long\def\Tdual{\T^{\delta,\star}}%
\global\long\def\omd#1{\omega_{#1}^{\delta}}%
\global\long\def\om#1{\omega_{#1}}%
\global\long\def\bul#1{{#1}^{\bullet}}%
\global\long\def\cir#1{{#1}^{\circ}}%

\global\long\def\Ob{\mathcal{O}}%
\global\long\def\Dom{\T}%
\global\long\def\Domepsi{\Omega_{\eps_{i}}}%
\global\long\def\Domp{\Omega\setminus\{v_{1},\dots,v_{n}\}}%
 
\global\long\def\Domr{\widetilde{\T}}%
\global\long\def\Surf{\widehat{\Omega}}%
\global\long\def\Surfeps{{{\widehat{\Omega}}_{\eps}}}%
\global\long\def\Surfo{\widehat{\Omega}_{0}}%
\global\long\def\bpoints{\{b_{1},\dots,b_{2k}\}}%
\global\long\def\SKDom{K_{\T_{\eps}}}%
 
\global\long\def\SKDr{K_{\T_{\eps}}^{\star}}%
 
\global\long\def\SKDomLim{K_{0}}%
 
\global\long\def\SKDrLim{K_{0}^{\star}}%
 
\global\long\def\SpStr{\mathfrak{c}}%
 
\global\long\def\Obs#1{\mathcal{O}\left[#1\right]}%
\global\long\def\hOp{\hat{\mathcal{O}}}%
 
\global\long\def\Ab{\mathcal{U}}%
 
\global\long\def\Abv{\mathcal{V}}%

\global\long\def\ds{\mathrm{ds}}%
 
\global\long\def\ns{\mathrm{ns}}%
 
\global\long\def\cs{\mathrm{cs}}%

\global\long\def\nord#1{:#1:}%
 
\global\long\def\Nexp#1#2{:e^{#1\Phi(#2)}:}%
 
\global\long\def\nexp#1#2{:e^{#1\varphi(#2)}:}%

\global\long\def\bsigma#1{\sqrt{2}:\cos\left(\frac{\sqrt{2}}{2}\Phi(#1)\right):}%
 
\global\long\def\bmu#1{\sqrt{2}:\sin\left(\frac{\sqrt{2}}{2}\Phi(#1)\right):}%
 
\global\long\def\ben#1{-\frac{1}{2} :|\nabla\Phi(#1)|^{2}:}%
 
\global\long\def\bpsi#1{2\sqrt{2}\i\partial\Phi(#1)}%
 
\global\long\def\bpstar#1{-2\sqrt{2}\i\bar{\partial}\Phi(#1)}%
 
\global\long\def\bpsipstar#1{e^{-\frac{1}{2}g_{\Omega}(#1,#1)}:\sin2\sqrt{\pi}\Phi(#1):}%

\global\long\def\setS{\mathcal{S}}%
 
\global\long\def\setSn{\mathcal{S}_{0}}%
 
\global\long\def\setSp{\mathcal{S}'_{0}}%
 
\global\long\def\gap{\alpha}%
 
\global\long\def\inv{\iota}%
\global\long\def\quadr#1#2{\mathbf{q}_{#1}(#2)}%
\global\long\def\hd{\mathcal{F}}%
\global\long\def\logmatrix{\mathcal{L}}%
\global\long\def\Ann{\mathcal{A}}%
\global\long\def\regint#1{\left(#1\right)_{\mathrm{reg}}}%
\global\long\def\ProdTh{\Pi_{v_{1},\dots,v_{n}}^{s'}}%
\global\long\def\vavg{\tilde{v}_{s'}}%

\title[Critical Ising correlations on a torus]{Critical Ising correlations on a torus}
\author{Baran Bayraktaroglu, Konstantin Izyurov}
\begin{abstract}
    We prove convergence of multi-point spin correlations in the critical Ising model on a torus. Via Pfaffian identities, this also implies convergence of other correlations, including correlations of spins with fermionic and energy observables. We obtain explicit formulae for the scaling limits in terms of theta functions, verifying the predictions in the physics literature. 
\end{abstract}

\maketitle
\tableofcontents{}

\section{Introduction}

The critical planar Ising model is a prototypical example of a lattice
model whose scaling limit is believed to be described by conformal
field theory \cite{BPZ}. A way to put these results on a concrete
ground is to stipulate that correlations of natural random variables
in the model, rescaled by a suitable power of the mesh size, converge
to conformally covariant limits, which can be then recognized as correlations
in the corresponding CFT. The power-law decay of correlations is a
manifestation of criticality, and conformal covariance of the limits
is a special feature of the critical models in dimension two.

Recently, there has been a lot of progress in rigorously establishing
the results of the above type in the case of the Ising model. In particular,
in \cite{CHI1}, convergence of spin correlations in the Ising model
on any simply-connected planar domains was proven, and in \cite{CHI2},
these results were extended to all primary fields (including energy,
disorder, and fermionic correlations). We will not attempt to
survey the extensive literature on the Ising model and its connections
to conformal field theory, referring the reader to the introduction
of \cite{CHI2}.

The goal of this paper is to extend the results obtained in \cite{CHI1,CHI2}
to the case of a flat torus. This has been already done in \cite{IKT}
for energy correlations. The main contribution of the present paper
is to compute the scaling limits of the spin and spin-fermion-fermion
correlations, which we do in Theorems \ref{thm: convergence of the observable}
and \ref{thm: conv_spins} below, and to confirm the explicit formulae
for the scaling limits predicted in the physics literature \cite{DiFSZ,DFMS}.
With that in hand, convergence of other correlations follow from combinatorial
Pfaffian identities exactly as in the planar case, as we explain in
Section \ref{subsec: all_correlations}. More precisely, we compute
correlations of spins (respectively, spins and two fermions) with
four ``topological'' observables $\mu^{(p,q)}$, $p,q\in\{0,1\}^{2}$,
satisfying $\sum_{p,q}\mu^{(p,q)}=1;$ one can thus get rid of these
observables by summing over $p,q.$ The need for introduction of these
observables is classical and goes back to early combinatorial solutions
of the Ising model \cite{Onsager}.

Let us highlight some of the motivation for studying the problem on
a torus. First, the full scope of conformal field theory (CFT) only manifests
itself when considered on Riemann surfaces. Hence, in order to rigorously
establish a connection between lattice models and CFT, we consider
it natural and important to move on from planar domains, and a torus
is the first case of interest, as evidenced by \cite{DiFSZ,DFMS}.
Second, the torus has been a classical and natural setup for studying
the Ising model, owing to the fact that various exact solvability
techniques work the best in the translation invariant setting, see
e.g. \cite{FF,Palmer,Salas, MW, Yang}; however, there seems to be no rigorous
treatment of (multi-point) correlation functions by these methods.
Third, the convergence results in the spirit of the present paper
are sometimes required to derive interesting consequences. One such
case is the work on sparse reconstruction problem in \cite{Galizka - Pete}. 

On the other hand, the case of the torus is in several respects easier
than the case of Riemann surfaces of arbitrary genus. First, a (flat)
torus has a ``canonical'' approximation by square lattice, which
is not necessarily the case for Riemann surfaces. This in particular
requires working, in higher genera, with more general families of graphs. Moreover, to
approximate a Riemann surface with a ``nice'' lattice, such as e.g.
isoradial graph, one has to endow it with a flat Euclidean structure;
when $g>1$, it will necessarily have conical singularities by Gauss-Bonnet
theorem. These conical singularities require a separate technical
work for convergence results. Also, in higher genera, certain degenerations
may occur (specifically, the analogs of Lemma \ref{lem: no_kernel}
below are not true in general), and more work is needed to bypass
this issue. In the torus case, the only degenerate case occurs in
the absence of spins, and it was treated in \cite{IKT}. In view of
these technicalities, we defer the treatment of higher genera to subsequent
work.

From a conceptual point of view, the scaling limits of correlations
in the Ising model, and in particular, fermionic observables, are
not functions but forms of fractional degree. In the case of a torus,
however, the distinction is not essential as we can work in the canonical
coordinate chart $z$, coming from obtaining a torus as a quotient
of the complex plane. We will therefore not always keep track of the
distinction.

The proof of convergence in the present paper mainly follows the arguments
in \cite{CHI2}. In Section \ref{subsec: definition_ and basic properties},
we set up the notation and define the main tool, the fermionic observable,
and establish its basic properties. This is the same observable as
in \cite{IKT}, but with spin insertions instead of energy insertions. A new challenge compared,
say, to the planar case, is the presence of the ``topological" term $\mu^{(p,q)},$
which is neither positive nor has any obvious monotonicity properties.
Therefore, \emph{a priori}, the normalization factor $\E(\sigma_{v_{1}}\dots\sigma_{v_{n}}\mu^{(p,q)})$
in the definition of the observable could vanish. In Section \ref{subsec: Non-vanishing},
we show that this in fact never happens, using a discrete complex
analysis argument (as mentioned above, such degeneracies \emph{could}
occur in higher genera). With that at hand, we proceed mostly as in
the planar case, computing the scaling limit of the fermionic observable
and then deducing the scaling limit of the logarithmic derivative
of $\E(\sigma_{v_{1}}\dots\sigma_{v_{n}}\mu^{(p,q)}).$ Fixing normalization
again requires a small modification compared to the planar case, again
due to lack of positivity/monotonicity properties of $\mu^{(p,q)}.$

We then proceed to the derivation of the explicit formulae for the
observables. In our approach to convergence, this is a non-trivial
step. (By contrast, an alternative approach to a proof of convergence via a combinatorial bosonization and convergence of dimer height function \cite{Dubedat, Hugoetc, Basok} yields explicit formulae more readily, but it has other technical drawbacks.) The main issue here is that the square
of the scaling limit of the fermionic observable is characterized
as a unique meromorphic differential with prescribed poles and some
restrictions on the residues, and \emph{zeros only of
even order}. This last conditions is difficult to check, even when
the candidate expression for the scaling limit is available. In \cite{BIVW},
the same problem was solved for planar multiply connected domains
using a limiting form of a classical identity (\ref{eq: bosonization})
due to Hejhal and Fay. Here, we adapt a similar argument, gluing a
second copy of the torus by thin ``handles'' at the positions of
spins, and then considering the limit of the Hejhal-Fay identity as
one pinches the handles. The details are different from \cite{BIVW}
due to the fact that there's no boundary, and pinching all the handles
disconnects the surface. We perform the relevant analysis in Section
\ref{sec: Szego} for surfaces of arbitrary genus, and specialize
to the torus case in Section \ref{sec: formulae_on_torus}. Let us mention that the exact expressions for the (squares of) scaling limits of correlations can be interpreted as correlations in a suitable compactified GFF, in the same vein as in \cite{BIVW}, and as anticipated by the combinatorial bosonization \cite{Dubedat} and the convergence of dimer height function to a compactified GFF on a torus \cite{Dubedat15}. 

\subsection{Acknowledgements}

The work was supported by Research Council of Finland grants
333932 and 363549. B. B. was supported by Research Council of Finland grant 357738. We are grateful
to Mikhail Basok for inspiring discussions. 

\medskip{}

\section{Discrete holomorphic fermionic observables on a torus}

\subsection{\label{subsec: definition_ and basic properties}Definitions and
basic properties of fermionic observables with spin insertions}

In this section, we introduce and establish some properties of the
discrete holomorphic fermionic observables on the torus. This is analogous
to \cite{CHI2}, where the construction was carried out in planar
domains, and \cite{IKT}, where the observables had energy insertions
rather than spin insertions. We will follow the conventions of \cite{CHI2},
in particular, we put $\dZ=\sqrt{2}e^{\frac{\i\pi}{4}}\delta\Z^{2},$
i.e., our square lattice is rotated by $45^{\circ}$ and the mesh
size $\delta$ denotes half of the diagonal. Note that the conventions
in \cite{IKT} are different.

We will focus on the critical Ising model on tori $\Td\coloneqq\dZ/\Lambdad$,
where $\delta>0$ is the mesh size, $\Lambdad$ is given by the lattice

\[
\Lambdad=\left\{ n\omd 1,+m\omd 2:n,m\in\Z\right\} 
\]
and $\omd 1,\omd 2\in\dZ$ are non-collinear. We will look at the
scaling limit of the model where $\delta\rightarrow0,\omd i\rightarrow\omega_{i}\in\C\setminus\{0\}$
for $i=1,2$. We will also assume that the modular parameter $\tau\coloneqq\omega_{2}/\omega_{1}$
satisfies $\im\tau>0.$

The Ising model on $\Td$ is the following probability measure on
spin configurations $\sigma:\Td\rightarrow\{\pm1\}$

\[
\P[\sigma]=\frac{1}{Z}\exp\left(\beta\sum_{x\sim y}\sigma_{x}\sigma_{y}\right)
\]
where the sum is over all pairs of nearest-neighboring vertices of
$\Td$, $\beta>0$ is a parameter called the \emph{inverse temperature},
and

\[
Z=\sum_{\sigma:\Td\rightarrow\{\pm1\}}\exp\left(\beta\sum_{x\sim y}\sigma_{x}\sigma_{y}\right)
\]
is the partition function of the model. We will fix the inverse temperature
to the critical value $\beta=\beta_{c}=\frac{1}{2}\log(\sqrt{2}+1)$.
We will denote by $\P$ and $\E$, respectively, the probability and
the expectation with respect to the above measure.

A \emph{corner} of the graph $\dZ$ is a midpoint of a segment joining
a lattice vertex $v\in\C^{\delta}$ with a vertex of the dual lattice
$\dZdual=\dZ+\delta$. Thus, the corners will form another square
lattice $\Cgr(\dZ)\coloneqq\frac{1}{2}\Cd+\frac{\delta}{2}$; the
corners are said to be \emph{adjacent} if they are nearest neighbors
on that square lattice. For the torus, the dual graph $\Tdual$ and
the corner graph $\Cgr(\Td)$ are defined as quotients of $\dZdual$
and $\Cgr(\dZ)$ by $\Lambda^{\delta}$. Given a point $z\in\Cgr(\Td)$,
we denote by $\cir z$ (respectively, $\bul z$) the vertex of $\Td$
(respectively, $\Tdual$) incident to $z$. 

We choose a \emph{Torelli marking} of the torus $\Td$, by letting
$\gamma_{10}$ and $\gamma_{01}$ be closed simple loops on $\Tdual$
that intersect at one vertex $b^{\bullet}$ and lift to paths on $\dZdual$
whose endpoints differ by $\omd 1$ and $\omd 2$, respectively. We
distinguish one of the corners $b$ adjacent to $b^{\bullet}$ and
call it \emph{the base corner}. We further denote $\gamma_{00}=\emptyset,\gamma_{11}=\gamma_{01}\cup\gamma_{10}$.
We introduce a 4--sheet cover $\Tdouble$ of $\Td$: 
\[
\Tdouble:=\{\dZ/2m\omd 1+2n\omd 2:m,n\in\Z\}.
\]
If $v_{1},\dots,v_{n}\in\Td$ (or more generally $\Td\cup\Tdual$),
we denote by $\Tdouble_{[v_{1},\dots,v_{n}]}$ the double cover of
$\Tdouble$ ramified at each point projecting to $v_{1},\dots,v_{n}$
and such that $2\gamma_{10}$ and $2\gamma_{01}$ lift to loops on
$\Tdouble_{[v_{1},\dots,v_{n}]}$. For $z\in\Tdouble_{[v_{1},\dots,v_{n}]},$
we denote by $z^{\star}$ the other fiber on $\Tdouble_{[v_{1},\dots,v_{n}]}$
over the same point -- that is, a small loop winding around any of
$v_{1},\dots,v_{n}$ lifts to a path on $\Tdouble_{[v_{1},\dots,v_{n}]}$
connecting $z$ to $z^{\star}.$ We denote the covering map from $\Tdouble$
to $\Td$ by $\pi$.

Given a corner $z\in\Cgr(\Tdouble_{[v_{1},\dots,v_{n}]})$, we denote
by $z_{pq}\in\Cgr(\Tdouble_{[v_{1},\dots,v_{n}]})$, $p,q\in\{0,1\}^{2}$,
a corner obtained by the following procedure: pick a simple path $\gamma$
on $\Cgr(\Td)$ from $\pi(z)$ to the base corner $b$ that does not
intersect either $\gamma_{01}$ or $\gamma_{10}$; then $z_{pq}$
is the endpoint of the lift of the loop $\gamma\cup\gamma_{pq}\cup\gamma^{-1}.$
(More precisely, by $\gamma_{pq}$ here we mean a path on $\Cgr(\Td)$
that follows $\gamma_{pq}$.) It is easy to see that $z_{pq}$ does
not depend on the choice of $\gamma$; note however that it does depend
on the relative positions of loops $\gamma_{01},\gamma_{10}$ and
punctures $v_{1},\dots,v_{n}.$

Given a subset $\gamma$ of edges of $\mathbb{T}^{\delta,\star}$,
we define the \emph{disorder observable} 
\[
\mu_{\gamma}:=e^{-2\beta\sum_{(xy)\cap\gamma\neq\emptyset}\sigma_{x}\sigma_{y}}
\]
where the sum is over the nearest neighbors $x\sim y\in\Td$ such
that the edge $(xy)$ intersects $\gamma$. We view subsets of edges
of $\Tdual$ as chains modulo 2 , in particular, $\gamma_{1}+\gamma_{2}=\gamma_{1}-\gamma_{2}$
will denote the symmetric difference of $\gamma_{1}$ and $\gamma_{2}$,
and $\partial\gamma$ is the boundary of $\gamma$, that is, the set
of all vertices of $\Tdual$ incident to an odd number of edges in
$\gamma$.

Given a corner $a\in\Cgr(\Td)$, we introduce the graph $\Cgr_{\{a\}}(\Td)$
by splitting the vertex $a$ in $\Cgr(\Td)$ into two vertices $a^{+}$
and $a^{-}$ along $[a^{\bullet}a^{\circ}]$, so that only paths on
$\Cgr(\Td)$ not crossing the segment $[a^{\bullet}a^{\circ}]$ are
allowed in $\Cgr_{\{a\}}(\Td)$, see \cite[Fig. 5.2]{IKT} or \cite[Remark 2.3]{CHI2}.
We similarly define $\Cgr_{\{a\}}(\Tdouble_{[v_{1},\dots,v_{n}]})$
and e.g. $\Cgr_{\{a_{1},\dots,a_{k}\}}(\Tdouble_{[v_{1},\dots,v_{n}]})$
with several split vertices. 

Let $\mu_{ij}$ be the disorder observable
corresponding to the path $\gamma_{ij}$, i.e. $\mu_{ij}=\mu_{\gamma_{ij}}$,
with $\mu_{00}=1$. We then define, for each $(p,q)\in\{0,1\}^{2}$,

\begin{equation}
\mu^{\sps}=\frac{1}{4}\sum_{i,j}(-1)^{\quadr{\sps}{i,j}}\mu_{ij},\label{eq: munu}
\end{equation}
where $\quadr{\sps}{i,j}=(1-p)i+(1-q)j+ij$ is (the quadratic form
corresponding to) one of the four \emph{spin structures} on $\Td$,
see \cite{Johnson, Cimasoni,Cim2,CCK} . 
\begin{defn}
\label{def: observable} Let $n$ be even. Given distinct vertices $\vv$ on the torus
$\Td$ and $v,w\in\Cgr(\Tdouble_{[v_{1},\dots,v_{n}]})$, we define
the discrete fermionic observable as

\begin{equation}
F_{v_{1},\dots,\v_{n}}^{(p,q)}(w,z)=\Corr{\Td}{\psi_{z}\psi_{w}\svv\mu^{(p,q)}}.\label{eq: fermobs}
\end{equation}
Here $\psi_{z}\psi_{w}=\drs z\drs w\sigma_{\cir z}\sigma_{\cir w}\mu_{\gamma}$,
where $\gamma$ is a subset of $\Tdual$ such that $\pa\gamma=\{\bul z,\bul w\}$,
and 
\begin{equation}
\drs z=e^{\frac{\i\pi}{4}}\left(\frac{\bul z-\cir z}{|\bul z-\cir z|}\right)^{-\frac{1}{2}}.
\end{equation}
\end{defn}

The expression in \eqref{eq: fermobs} depends on some choices, namely,
the choice of the disorder line $\gamma$ and the branches of the
square roots of $\eta_{z},\eta_{w}$. However, as explained in \cite[Section 2]{CHI2},
there is a natural way to connect these choices as $z$ moves to its
nearest neighbor $\hat{z}$ on $\Cgr_{\{w\}}(\Td).$ Namely, the square root in $\eta_{\hat{z}}$ is extended from that in
$\eta_{z}$ ``by continuity'' (i.e., we choose is so that $|\arg(\eta_{\hat{z}}/\eta_z)|=\frac{\pi}{8}$); if $z^{\bullet}\neq \hat{z}^{\bullet}]$, then $\gamma$ is amended by adding the segment
$[z^{\bullet}\hat{z}^{\bullet}]$ modulo two, and if $z^{\circ}\neq \hat{z}^{\circ}$, we add a $-$ sign
whenever the segment $[z^{\circ}\hat{z}^{\circ}]$ crosses $\gamma$; note that for any nearest neighbor move $z\to \hat{z}$ on $\Cgr_{\{w\}}(\Td),$ exactly one of the above alternatives happens.

With this convention, we have the following Lemma: 
\begin{lem}
\label{lem: properties of F} Let $w\in\Cgr(\Tdouble_{[v_{1},\dots,v_{n}]}).$
Then, assuming the initial condition 
\[
F_{v_{1},\dots,v_{n}}^{(p,q)}(w,w^{+})=\eta_{w}^{2}\E(\sigma_{v_{1}}\dots\sigma_{v_{n}}\mu^{(p,q)}),
\]
the observable $\FFS{w,\cdot}{\vv}{\sps}$ is well defined on $\Cgr_{\{w\}}(\Tdouble_{[v_{1},\dots,v_{n}]}),$
and \textup{$\eta_{w}^{-1}\FFS{w,\cdot}{\vv}{\sps}$} is s-holomorphic.
Moreover, one has the spinor property 
\begin{equation}
\FFS{w,z}{\vv}{\sps}=-\FFS{w,z^{\star}}{\vv}{\sps}\label{eq: F_spinor}
\end{equation}
anti-periodicity condition 
\begin{equation}
\FFS{w,z_{kl}}{\vv}{\sps}=(-1)^{pk+ql}\FFS{w,z}{\vv}{\sps}\label{eq: F_antiperiodic}
\end{equation}
and has special values 
\begin{equation}
\FFS{w,w^{\pm}}{\vv}{\sps}=\pm\eta_{w}^{2}\E_{\Td}[\sigma_{v_{1}}\dots\sigma_{v_{n}}\mu^{\sps}].\label{eq: F_value_at_w}
\end{equation}
Finally, if $w\sim v_{1}$ and $z\neq w$ is such that $z^{\bullet}=w^{\bullet},$
and the segment $[z^{\circ}v_{1}]$ does not intersect any disorder
lines in the expansion of $\mu^{\sps}$, we have 
\begin{equation}
\FFS{w,z}{\vv}{\sps}=\eta_{z}\eta_{w}\E_{\Td}[\sigma_{z^{\circ}}\sigma_{v_{2}}\dots\sigma_{v_{n}}\mu^{\sps}],\label{eq: move_spin_one_vertex}
\end{equation}
 where for the sign of $\eta_{z},$ we choose $\eta_{z}=\eta_{w}e^{-\frac{\i\theta}{2}},$
where $\theta$ is a counterclockwise angle between the segments $[z^{\circ}-z^{\bullet}]$
and $[w^{\circ}-w^{\bullet}].$
\end{lem}

\begin{proof}
The s-holomorphicity and the properties (\ref{eq: F_spinor}), (\ref{eq: F_value_at_w}),
(\ref{eq: move_spin_one_vertex}) have local nature and thus their
proofs are exactly as in the planar case, see \cite[Section 2.3]{CHI2}
and the proof of \cite[Corollary 5.2]{CHI2}. Hence, it only remains
to prove (\ref{eq: F_antiperiodic}). We first note the following
transformation rule: 
\begin{equation}
\E[\psi_{z_{kl}}\psi_{w}\sigma_{v_{1}}\dots\sigma_{v_{n}}\mu_{ij}]=(-1)^{kl+k+l}(-1)^{li+kj}\E[\psi_{z}\psi_{w}\sigma_{v_{1}}\dots\sigma_{v_{n}}\mu_{i+k,j+l}].\label{eq: mu_ij_transofrm_rule}
\end{equation}
To check (\ref{eq: mu_ij_transofrm_rule}), recall that $z_{kl}$
is defined by moving $z$ to the base corner $b$, tracing $\gamma_{kl},$
and then moving back to $z$ along the same path. We assume for without
loss of generality that $z$ is not adjacent to either $\gamma_{10}$
or $\gamma_{10}$, and that $z$ stays on the same side of $\gamma_{10}$
or $\gamma_{01}$ as we move. When moving $z$ along $\gamma_{10}$,
$z^{\bullet}$ creates (or erases) a disorder line $\gamma_{10}$,
and similarly for moving along $\gamma_{01},$ which explains the
change of $\mu_{ij}$ into $\mu_{i+k,j+l}.$ It is straightforward
to check that the $\eta_{z}$ factor does not change; hence, we only
need to collect the signs from $z^{\circ}$ crossing disorder lines.
The factor $(-1)^{li}$ (resp., $(-1)^{kj}$) is from crossing $\gamma_{10}$
as we move along $\gamma_{01}$ (resp., vice versa). The factor $(-1)^{kl}$
accounts for the fact that if we do both sequentially, when tracing
$\gamma_{01},$ $z^{\circ}$ also jumps across a disorder line \emph{created
(or erased) when tracing $\gamma_{10}$. }Finally, the factors $(-1)^{k}$
and $(-1)^{l}$ come from the fact that when tracing $\gamma_{10}$
and $\gamma_{01}$, we will cross the disorder line that was created
as $z$ moved to $b$ (and then is subsequently erased as we move
back). Thus, (\ref{eq: mu_ij_transofrm_rule}) is proved, and we conclude
that \emph{
\begin{multline*}
\FFS{w,z_{kl}}{\vv}{\sps}=\frac{1}{4}\sum_{i,j}(-1)^{\quadr{\sps}{i,j}+kl+li+kj+k+l}\E[\psi_{z}\psi_{w}\sigma_{v_{1}}\dots\sigma_{v_{n}}\mu_{i+k,j+l}]\\
=\frac{1}{4}(-1)^{pk+ql}\sum_{i,j}(-1)^{\quadr{\sps}{i+k,j+l}}\E[\psi_{z}\psi_{w}\sigma_{v_{1}}\dots\sigma_{v_{n}}\mu_{i+k,j+l}]=\frac{1}{4}(-1)^{pk+ql}\FFS{w,z}{\vv}{\sps},
\end{multline*}
}since 
\begin{multline*}
\quadr{\sps}{i,j}+kl+li+kj=(i+k)(j+l)+(1-p)i+(1-q)j\\
=\quadr{\sps}{i+k,j+l}-(1-p)k-(1-q)l.
\end{multline*}
\end{proof}

\subsection{\label{subsec: Non-vanishing}Non-vanishing of the discrete residue.}

In order to pass to the scaling limit, one needs to normalize the
observable by its ``discrete residues'' $\FFS{w,w^{+}}{\vv}{\sps}-\FFS{w,w^{-}}{\vv}{\sps}.$
Here we will show that this ``discrete residue'', given by (\ref{eq: F_value_at_w}),
never vanishes unless $n=0,$ $(p,q)=(0,0).$ That case required a
separate treatment which was carried out in \cite{IKT}. We recall
the classical procedure \cite{Smirnov} for integrating the square
of an s-holomorphic function:
\begin{defn}
\label{defn: intfsq} Given an s-holomorphic function $F$ on $\Td$,
we define the real-valued, additivaly multi-valued function on $\Td\cup\Tdual$ 
\[
G(w)=\int^{w}\im(F(z)^{2}dz)
\]
to be the integral of the (closed) discrete differential form $-2\i F(z)^{2}(\bul z-\cir z)$,
in other words, for neighboring $\bul z\in\Tdual$ and $\cir z\in\Td$,
we have

\[
G(\bul z)-G(\cir z)=-2\i F(z)^{2}(\bul z-\cir z)=2\delta|F(z)|^{2}.
\]
\end{defn}

\begin{lem}
\label{lem: no_kernel}Let $v_{1},\dots,v_{n}\in\T^{\delta}$, where
$n>0,$ and let $F$ be an s-holomorphic function on $\Cgr\left(\Tdouble_{[v_{1},\dots,v_{n}]}\right)$
satisfying $F(z)=-F(z^{\star})$ and $F(z_{ij})=(-1)^{pi+qj}F(z)$
for some $p,q\in\{0,1\}.$ Then, $F\equiv0.$
\end{lem}

\begin{rem}
A continuous version of this Lemma has a short proof, given in the proof of Theorem \ref{thm: convergence of the observable} below. In the discrete, we need a more subtle argument since a square of discrete holomorphic function is not necessarily discrete holomorphic. Note that the punchline of either proof is, essentially, that a (discrete) holomorphic differential on a torus cannot have simple zeros; this would fail in higher genera.  
\end{rem}

\begin{proof}
Let $G$ be the the imaginary part of the discrete integral of $F^{2}$,
which we consider restricted to $\Tdual$. Then, $G$ is an additively
multi-valued function, in particular, its discrete derivatives $G(u')-G(u)$
for $u'\sim u$ are well defined on $\T^{\delta,\star}.$ It therefore
follows from the discrete Green's identity that 
\[
\sum_{u\in\Tdual}\Delta^{\bullet}G(u)=0.
\]
On the other hand, recall that we have 
\[
\frac{1}{\sqrt{2}\delta}\Delta^{\bullet}G(u)=2\left|\sum_{z\sim u}\eta_{z}^{2}F(z)\right|,
\]
(see the proof of Lemma 1.5 in the Appendix of \cite{IzyurovFree}
and \cite[Lemma 4.8]{CHI2}), in particular, $\Delta^{\bullet}G(u)\geq0$
for all $u\in\Tdual.$ Hence, $G$ is discrete harmonic, and 
\begin{equation}
\sum_{z\sim u}\eta_{z}^{2}F(z)=0\label{eq: disc_lap_vanishes}
\end{equation}
 for each $u\in\Tdual.$ Denote the corners neighboring to $u$ by
$z_{1},z_{2},z_{3},z_{4}$ (in counterclockwise order), and let $e_{12},e_{34}$
be the edges of $\Tdual$ incident to $z_{1},z_{2}$ (respectively,
$z_{3},z_{4}$). For $e$ an edge of $\Tdual$, denote by $F(e)$
the unique complex number such that 
\[
\frac{1}{2}\left(F(e)+\eta_{z}^{2}\bar{F(e)}\right)=\proj{\eta_{z}}{F(e)}=F(z).
\]
 for each corner $z\sim e;$ recall that the existence of such a number
if equivalent to s-holomorphicity. Plugging this identity into
(\ref{eq: disc_lap_vanishes}) yields 
\[
\left(\eta_{z_{1}}^{2}+\eta_{z_{2}}^{2}\right)F(e_{12})+\left(\eta_{z_{1}}^{4}+\eta_{z_{2}}^{4}\right)\bar{F(e_{12})}+\left(\eta_{z_{3}}^{2}+\eta_{z_{4}}^{2}\right)F(e_{34})+\left(\eta_{z_{3}}^{4}+\eta_{z_{4}}^{4}\right)\bar{F(e_{34})}=0.
\]
Noting that $\eta_{z_{4}}^{4}=-\eta_{z_{3}}^{4}=\eta_{z_{2}}^{4}=-\eta_{z_{1}}^{4}$
and $\eta_{z_{3}}^{2}+\eta_{z_{4}}^{2}=-\eta_{z_{1}}^{2}-\eta_{z_{2}}^{2}\neq0,$
this yields 
\[
F(e_{12})=F(e_{34}),
\]
 i.e., $F$ is constant on each line running along edges of $\Tdual.$
Now, let $e_{23}$ be the edge incident to $z_{2},z_{3}.$ Since $\proj{\eta_{z_{2}}}{F(e_{23})}=\proj{\eta_{z_{2}}}{F(e_{12})}$
and $\proj{\eta_{z_{3}}}{F(e_{23})}=\proj{\eta_{z_{3}}}{F(e_{34})}=\proj{\eta_{z_{3}}}{F(e_{12})},$
we conclude that $F(e_{23})=F(e_{12}),$ i.e., the constant values
of $F$ coincide on the perpendicular lines running along edges. We
conclude that $F(z)\equiv\proj{\eta_{z}}c$ for a constant $c\in\C$,
but if $n>0$ the only such constant consistent with the condition
$F(z)=-F(z^{\star})$ is zero. 
\end{proof}

We now examine the properties of the quantity $\Corr{\Td}{\svv\mu^{(p,l)}}$
\emph{as a function of the positions of $v_{1},\dots,v_{n},$ }under
the rule that a spin-disorder correlation changes sign when a spin
jumps across a disorder line.
\begin{lem}
\label{lem: shift_denom} For $p,q,k,l\in\{0,1\}$, one has 
\[
\Corr{\Td}{\sigma_{(v_{1})_{kl}}\sigma_{v_{2}}\dots\sigma_{v_{n}}\mu^{\sps}}=\Corr{\Td}{\svv\mu^{(p+l,q+k)}}.
\]
where $(p+l,q+k)$ is to be taken $\mod2$. 
\end{lem}

\begin{proof}
Recall that 

\[
\Corr{\Td}{\svv\mu^{(p+k,l+q)}}=\frac{1}{4}\sum_{i,j}(-1)^{q_{(p,q)}(i,j)}\Corr{\Td}{\svv\mu_{ij}}
\]
We now note that $\Corr{\Td}{\sigma_{(v_{1})_{kl}}\sigma_{v_{2}}\dots\sigma_{v_{n}}\mu_{ij}}=(-1)^{kj+li},$
since as we move $v_{1}$ around the torus along $\gamma_{kl}$, we
add a minus sign whenever we are crossing $\gamma_{ij}$. Since we
have 
\[
(-1)^{q_{\sps}(i,j)}(-1)^{kj+li}=(-1){}^{(1-p+l)i+(1-q+k)j+ij}=(-1)^{q_{(p+l,q+k)}(i,j)},
\]
 the result follows.
\end{proof}
\begin{lem}
\label{lem: denom_does_not_vanish}For any $\vv$ with $n>0$ even
and any $p,q$, we have 
\[
\Corr{\Td}{\svv\mu^{\sps}}\neq0.
\]
\end{lem}

\begin{proof}
Assume the contrary. Then, for any $w\in\Cgr(\Tdouble),$ we have
by \ref{eq: F_value_at_w} $\FFS{w,w^{\pm}}{\vv}{\sps}=0,$ that is,
$\FFS{w,\cdot}{\vv}{\sps}$ is in fact well-defined (and s-holomorphic)
on $\Cgr(\Tdouble)$ rather than $\Cgr_{\{w\}}(\Tdouble).$ By Lemma
\ref{lem: properties of F}, this function satisfies all the conditions
of Lemma \ref{lem: no_kernel} and thus it is identically zero. By
choosing $w\sim v_{1},$ we now get by (\ref{eq: move_spin_one_vertex})
that 
\[
\E_{\Td}[\sigma_{\hat{v}_{1}}\sigma_{v_{2}}\dots\sigma_{v_{n}}\mu^{\sps}]=0
\]
for each $\hat{v}_{1}\in\Td$ neighboring $v_{1}.$ We can thus move
$v_{1}$ further while the identity still holds true; dragging it
around the torus we obtain by Lemma \ref{lem: shift_denom} that 
\[
\E_{\Td}[\sigma_{v_{1}}\sigma_{v_{2}}\dots\sigma_{v_{n}}\mu^{(\hat{p},\hat{q})}]=0.
\]
for each $(\hat{p},\hat{q}).$ But now, noting that 
\begin{equation}
1=\mu^{(0,0)}+\mu^{(0,1)}+\mu^{(1,0)}+\mu^{(1,1)},\label{eq: sum_mus}
\end{equation}
 we deduce that $\Corr{\Td}{\svv}=0,$ which is a contradiction since
for even $n$, this correlation is always positive.
\end{proof}

\subsection{Scaling limit of spin correlations}

We are now in a position to prove the results concerning scaling limits
of $\FFS{w,z}{\vv}{\sps},$ and deduce the convergence of the spin
correlations. Recall that $\omega_{1}^{\delta},\omega_{2}^{\delta}\in\Cd$
are the periods used to define the discrete torus $\Td.$ Given $\omega_{1},\omega_{2}\in\C$
with $\im(\omega_{2}/\omega_{1})>0,$ we define $\T,\hat{\T},$ and
$\hat{\T}_{[v_{1},\dots,v_{n}]}$ similarly to their discrete counterparts. Throughout this section, we assume $n$ to be even.
\begin{thm}
\label{thm: convergence of the observable}Let $n>0.$ Assume that
$\omd i\rightarrow\omega_{i}\in\C\setminus\{0\},$ $i=1,2,$ as \textup{$\delta\rightarrow0$,
and that }$\tau\coloneqq\omega_{2}/\omega_{1}$ satisfies $\im\tau>0.$
Then, one has 
\begin{equation}
\frac{\bar{\eta}_{w}\FFS{w,z}{\vv}{(p,q)}}{\delta\Corr{\Td}{\svv\mu^{\sps}}}=C_{\psi}^{2}\proj{\emph{\ensuremath{\drs z}}}{\ffs{w,z}{\vv}{[\eta_{w}],(p,q)}}+o(1)\label{eq: convergence}
\end{equation}
uniformly as long as $v_{1},\dots,v_{n},w,z$ stay at least a fixed
distance from each other. Here $C_{\psi}=\left(\frac{2}{\pi}\right)^{\frac{1}{2}},$
$\ffs{w,\cdot}{\vv}{[\eta],(p,q)}$ is a holomorphic function on $\hat{\T}_{[v_{1}.\dots,v_{n}]}\setminus\{w_{ij},w_{ij}^{\star}:i,j\in\{0,1\}\}$
satisfying the following properties: 
\begin{equation}
\ffs{w,z_{ij}}{\vv}{[\eta],(p,q)}=(-1)^{pi+qj}\ffs{w,z}{\vv}{[\eta],(p,q)}\quad\text{and}\quad\ffs{w,z^{\star}}{\vv}{[\eta],(p,q)}=-\ffs{w,z}{\vv}{[\eta],(p,q)}\label{eq: contin_period}
\end{equation}
\end{thm}

\begin{equation}
\ffs{w,z}{\vv}{[\eta],(p,q)}=\frac{\bar{\eta}}{z-w}+o(1),\quad z\rightarrow w\label{eq: contin_pole}
\end{equation}
\begin{equation}
\ffs{w,z}{\vv}{[\eta],(p,q)}=\frac{e^{\i\pi/4}c_{i}}{\sqrt{z-v_{i}}}+o(1),\quad z\rightarrow v_{i}\label{eq: contin_sing}
\end{equation}
for some $c_{i}=c_{i}(\omega_{1},\omega_{2},v_{1},\dots,v_{n},w,p,q)\in\R.$
Moreover, the function $\ffs{w,\cdot}{\vv}{[\eta],(p,q)}$ is uniquely
determined by the above conditions. 
\begin{proof}
We start with the ``moreover'' claim, which is just a continuous
analog of Lemma~\ref{lem: no_kernel}. Assuming that such a function
is not unique, we can subtract them to get a function $g$ holomorphic
on $\hat{\T}_{[v_{1}.\dots,v_{n}]}$ and satisfying \ref{eq: contin_period}
and \ref{eq: contin_sing}. Then, $g^{2}$ is a meromorphic function
on $\T$ with simple poles at $v_{1},\dots,v_{n}$, and thus 
\[
0=-\i\sum_{j=1}^{n}\res_{v_{j}}g^{2}=\sum_{j=1}^{n}c_{j}^{2},
\]
 from which it follows that $g^{2}$ is in fact holomorphic on $\T$,
and thus $g$ is constant. However, the only constant compatible with
(\ref{eq: contin_period}) is zero.

With that at hand, the rest of the proof is similar to \cite[Theorem 5.30]{CHI2},
with simplifications stemming from absence of the boundary, so we
only sketch it here. Let $\tilde{F}_{\delta}(z)$ denote the left-hand
side of \ref{eq: convergence}, and let 
\[
M_{\eps}^{\delta}=\sup\{|\tilde{F}_{\delta}(z)|:\dist(z,\{v_{1},\dots,v_{n},w\})\geq\eps\}.
\]
If $M_{\eps}^{\delta}$ is bounded for every fixed $\eps$, then the
usual theory of pre-compactness of s-holomorphic functions imply that
along a subsequence,
\[
\tilde{F}_{\delta_{k}}(z)=\proj{\eta_{z}}{f(z)}+o(1)
\]
uniformly in $z$ at distance at least $\eps$ from the other marked
points, where $f$ is holomorphic in $\hat{\T}_{[v_{1},\dots,v_{n}]}\setminus \{w_{ij},w_{ij}^{\star}:i,j\in\{0,1\}\}.$
The property (\ref{eq: contin_period}) for $f$ follows from the
corresponding properties (\ref{eq: F_spinor}--\ref{eq: F_antiperiodic})
of $\tilde{F}.$ The properties (\ref{eq: contin_pole}) and (\ref{eq: contin_sing})
are purely local consequences of the properties of $\tilde{F}$ in
the discrete, and thus the same proof as in the planar case applies.
The uniqueness of $f$ established above proves the desired convergence. 

Assuming that $M_{\eps}^{\delta}\to\infty$ for some $\eps>0$, we
consider $\check{F}_{\delta}(\cdot)=\left(M_{\eps}^{\delta}\right)^{-1}\tilde{F}_{\delta}(\cdot),$
and notice that now 
\[
\sup\{|\tilde{F}_{\delta}(z)|:\dist(z,\{v_{1},\dots,v_{n},w\})\geq\eps'\}
\]
stays bounded as $\delta\to0$ for \emph{every} $\eps'>0.$ Indeed,
we may assume that $\eps$ is so small that the discs $B_{\eps}(v_{1}),\dots B_{\eps}(v_{n}),B_{\eps}(w)$
do not overlap, and then the statement can be proven separately for
each annulus $B_{\eps}(v_{i})\setminus B_{\eps'}(v_{i})$ and $B_{\eps}(w)\setminus B_{\eps'}(w)$
(since it is trivially true for $\eps=\eps'$). These are again local
statements and hence the planar argument applies, see \cite{CHI2}.
We conclude as above that $\check{F}_{\delta}(\cdot)$ converges to
a holomorphic function $\check{f}$ satisfying (\ref{eq: contin_period}--\ref{eq: contin_sing})
with $\eta=0,$ which by the above argument is identically zero. But
then 
\[
\sup\{|\check{F}_{\delta}(\cdot)|:\dist(z,\{v_{1},\dots,v_{n},w\})\geq\eps\}<\frac{1}{2}
\]
 for $\delta$ small enough, a contradiction. 
\end{proof}
We remark that the existence of the function $\ffs{w,z}{\vv}{[\eta],(p,q)}$
with the above properties follows from the convergence result above;
we will give an independent proof in the next section as we derive
an explicit formula. Note that $\eta\mapsto\ffs{w,z}{\vv}{[\eta],(p,q)}$
is real linear in $\eta.$ We put 
\[
\ffs{w,z}{\vv}{(p,q)}=\ffs{w,z}{\vv}{[1],(p,q)}+\i\ffs{w,z}{\vv}{[\i],(p,q)};
\]
\[
\ffs{w,z}{\vv}{\star,(p,q)}=\ffs{w,z}{\vv}{[1],(p,q)}-\i\ffs{w,z}{\vv}{[\i],(p,q)}.
\]
The following result is similar to \cite[Lemma 3.17]{CHI2}:
\begin{lem}
The two functions $\ffs{w,z}{\vv}{(p,q)},\ffs{w,z}{\vv}{\star,(p,q)}$
are holomorphic in their second argument, satisfy (\ref{eq: contin_period}),
and have the anti-symmetry property 
\[
\ffs{w,z}{\vv}{(p,q)}=-\ffs{z,w}{\vv}{(p,q)};\quad\ffs{w,z}{\vv}{\star,(p,q)}=-\overline{\ffs{z,w}{\vv}{\star,(p,q)}}.
\]
We also have the asymptotics $\ffs{w,z}{\vv}{(p,q)}\sim2(z-w)^{-1}+o(1),$
$\ffs{w,z}{\vv}{\star,(p,q)}=O(1)$ as $z\to w,$ and the relation
\[
\ffs{w,z}{\vv}{[\eta],(p,q)}=\frac{1}{2}\bar{\eta}\ffs{w,z}{\vv}{(p,q)}+\frac{1}{2}\eta\ffs{w,z}{\vv}{\star,(p,q)}.
\]
\end{lem}

\begin{proof}
Given $w_{1},w_{2}\in\hat{\T}_{[v_{1}.\dots,v_{n}]}$ and $\eta_{1},\eta_{2}\in\C\setminus\{0\},$
we can write 
\[
0=\re\sum\res\left(\ffs{w_{1},\cdot}{\vv}{[\eta_{1}],(p,q)}\ffs{w_{2},\cdot}{\vv}{[\eta_{2}],(p,q)}\right)=\re\left[\bar{\eta}_{1}\ffs{w_{2},\cdot}{\vv}{[\eta_{2}],(p,q)}+\bar{\eta}_{2}\ffs{w_{2},\cdot}{\vv}{[\eta_{1}],(p,q)}\right],
\]
since the product $\ffs{w_{1},\cdot}{\vv}{[\eta_{1}],(p,q)}\ffs{w_{2},\cdot}{\vv}{[\eta_{2}],(p,q)}$
is a meromorphic function on $\T$ by (\ref{eq: contin_period}),
and by (\ref{eq: contin_sing}), its residues at $v_{i}$ are purely
imaginary. With this identity, which is analogous to \cite[Lemma 3.15]{CHI2},
the rest of the proof is exactly as in \cite[Lemma 3.17]{CHI2}.
\end{proof}
\begin{cor}
\label{cor:log_derivative}If $v_{1},\hat{v}_{1}\in\Td$ are adjacent
and the edge $(v_{1},\hat{v}_{1})$ does not cross the disorder lines
involved in the definition of $\mu^{(p,q)}$, we have as $\delta\to0$
\[
\frac{\E_{\T_{\delta}}(\sigma_{\hat{v}_{1}}\sigma_{v_{2}}\dots\sigma_{v_{n}}\mu^{(p,q)})}{\E_{\T_{\delta}}(\sigma_{v_{1}}\sigma_{v_{2}}\dots\sigma_{v_{n}}\mu^{(p,q)})}=1+\re\left[\mathcal{A}_{[v_{1},\dots,v_{n}]}^{(p,q)}\cdot(\hat{v}_{1}-v_{2})\right]+o(\delta),
\]
uniformly as long as $v_{1},\dots,v_{n}$ stay at distance at least
$\eps>0$ from each other. Here the coefficient $\mathcal{A}_{[v_{1},\dots,v_{n}]}^{(p,q)}$
is defined by 
\begin{equation}
2\mathcal{A}_{[v_{1},\dots,v_{n}]}^{(p,q)}=\lim_{z\to v_{1}}(z-v_{1})^{-1}\left[(z-v_{1})^{\frac{1}{2}}\lim_{w\to v_{1}}(w-v_{1})^{\frac{1}{2}}\ffs{w,z}{\vv}{(p,q)}-1\right].\label{eq: def_A}
\end{equation}
\end{cor}

\begin{proof}
The proof of the analogous result in the planar case \cite[Theorem 3.39]{CHI2},
is purely local: it derives this result from the discrete holomorphicity
the observable $\FFS{w,z}{\vv}{(p,q)}$ with respect to both arguments,
and the convergence result of \cite[Theorem 3.16]{CHI2} applied in
the vicinity of the point $v_{1}.$ With that convergence result replaced
by Theorem \ref{thm: convergence of the observable}, the rest of
the proof applies without change.
\end{proof}
\begin{cor}
\label{cor: conv_ratio}Given $\eps>0,$ we have, for all $v_{1},\dots,v_{n}$
at distance at least $\eps$ from each other and all $\hat{v}_{1},\dots,\hat{v}_{n}$
at distance at least $\eps$ from each other,
\begin{equation}
\frac{\E_{\T_{\delta}}(\sigma_{v_{1}}\dots\sigma_{v_{n}}\mu^{(p,q)})}{\E_{\T_{\delta}}(\sigma_{\hat{v}_{1}}\dots\sigma_{\hat{v}_{n}}\mu^{(p,q)})}=\exp\left(R_{\hat{v}_{1},\dots,\hat{v}_{n}}^{(p,q)}(v_{1},\dots,v_{n})+o(1)\right)\label{eq: conv_ratio_vs}
\end{equation}
uniformly as $\delta\to 0$. Here $R_{\hat{v}_{1},\dots,\hat{v}_{n}}^{(p,q)}$ is defined
on the simply-connected sheet obtained by cutting $\T$ along $\gamma_{01}$
and $\gamma_{10}$ by the conditions 
\begin{equation}
dR_{\hat{v}_{1},\dots,\hat{v}_{n}}^{(p,q)}=\re\sum_{i=1}^{n}\left(\mathcal{A}_{[v_{i},v_{1},\dots,v_{i-1},v_{i+1},\dots,v_{n}]}^{(p,q)}\cdot dv_{i}\right)\label{eq: dR}
\end{equation}
and $R_{\hat{v}_{1},\dots,\hat{v}_{n}}^{(p,q)}(\hat{v}_{1},\dots,\hat{v}_{n})=0.$
\end{cor}

\begin{proof}
Follows by integrating the result of Corollary \ref{cor:log_derivative}
as in \cite[Corollary 5.2]{CHI2}.
\end{proof}
We now note that in view of Lemma \ref{lem: shift_denom}, we can
naturally extend $R_{\hat{v}_{1},\dots,\hat{v}_{n}}^{(p,q)}$ to a
function defined on $\hat{\T}^{\times n}\setminus\left\{\exists i\neq j:\pi(v_{i})=\pi(v_{j})\right\}$
(recall that $\pi$ denotes the covering map from $\hat{\T}$
to $\T$), with the convergence result extended accordingly. In particular,
we have 
\[
\frac{\E_{\T_{\delta}}(\sigma_{v_{1}}\dots\sigma_{v_{n}}\mu^{(p+l,q+k)})}{\E_{\T_{\delta}}(\sigma_{v_{1}}\dots\sigma_{v_{n}}\mu^{(p,q)})}=\exp\left(R_{v_{1},\dots,v_{n}}^{(p,q)}((v_{1})_{kl},\dots,v_{n})+o(1)\right).
\]
Combining this identity with (\ref{eq: sum_mus}), we see that the
tasks of identifying the scaling limit of $\E_{\T_{\delta}}(\sigma_{v_{1}}\dots\sigma_{v_{n}})$
and of each of $\E_{\T_{\delta}}(\sigma_{v_{1}}\dots\sigma_{v_{n}}\mu^{(p,q)})$
are equivalent. We introduce the purported scaling limit by the following
properties:
\begin{defn}
We define $\ccor{\sigma_{v_{1}}\dots\sigma_{v_{n}}\mu^{(p,q)}}_{\T},$ $p,q\in\{0,1\},$
to be the four functions on $\hat{\T}^{\times n}\setminus\left\{\exists i\neq j:\pi(v_{i})=\pi(v_{j})\right\}$
satisfying $\ccor{\sigma_{(v_{1})_{kl}}\dots\sigma_{v_{n}}\mu^{(p,q)}}_{\T}=\ccor{\sigma_{v_{1}}\dots\sigma_{v_{n}}\mu^{(p+l,q+k)}}_{\T},$
\begin{equation}
\frac{\ccor{\sigma_{v_{1}}\dots\sigma_{v_{n}}\mu^{(p,q)}}}{\ccor{\sigma_{\hat{v}_{1}}\dots\sigma_{\hat{v}_{n}}\mu^{(p,q)}}}=\exp\left(R_{\hat{v}_{1},\dots,\hat{v}_{n}}^{(p,q)}(v_{1},\dots,v_{n})\right),\label{eq: cont_ratio}
\end{equation}
 and the normalization condition 
\begin{equation}
\lim_{v_{2}\to v_{1},\dots,v_{n}\to v_{n-1}}\frac{\ccor{\sigma_{v_{1}}\dots\sigma_{v_{n}}}}{|v_{2}-v_{1}|^{-\frac{1}{4}}\dots|v_{n}-v_{n-1}|^{-\frac{1}{4}}}=1,\label{eq: merge_asymptotics}
\end{equation}
where $\ccor{\sigma_{v_{1}}\dots\sigma_{v_{n}}}=\sum_{p,q}\ccor{\sigma_{v_{1}}\dots\sigma_{v_{n}}\mu^{(p,q)}}=\sum_{k,l}\ccor{\sigma_{(v_{1})_{kl}}\dots\sigma_{v_{n}}\mu^{(0,0)}}.$
\end{defn}

Put $C_{\sigma}=2^{\frac{1}{6}}e^{\frac{3}{2}\zeta'(-1)}.$ We are
in the position to state the convergence theorem: 
\begin{thm}
\label{thm: conv_spins}The functions $\ccor{\sigma_{v_{1}}\dots\sigma_{v_{n}}\mu^{(p,q)}}_{\T}$
as above exist, are unique, and we have, for every $\eps>0$ and every
even $n,$ 
\[
\delta^{-\frac{n}{8}}\E_{\T_{\delta}}(\sigma_{v_{1}}\dots\sigma_{v_{n}}\mu^{(p,q)})=C_{\sigma}^{n}\ccor{\sigma_{v_{1}}\dots\sigma_{v_{n}}\mu^{(p,q)}}_{\T}+o(1),\quad\delta\to0,
\]
 uniformly in $v_{1},\dots,v_{n}$ at distance at least $\eps$ from
each other. Moreover, 
\[
\delta^{-\frac{n}{8}}\E_{\T_{\delta}}(\sigma_{v_{1}}\dots\sigma_{v_{n}})=C_{\sigma}^{n}\ccor{\sigma_{v_{1}}\dots\sigma_{v_{n}}}_{\T}+o(1),\quad\delta\to0.
\]
\end{thm}

\begin{proof}
It is clear that the function satisfying (\ref{eq: cont_ratio}) is
unique up to multiplicative constant, and the constant is uniquely
fixed by the asymptotics (\ref{eq: merge_asymptotics}). We can write
\begin{multline*}
\frac{\E_{\T_{\delta}}(\sigma_{v_{1}}\dots\sigma_{v_{n}})}{\E_{\T_{\delta}}(\sigma_{\hat{v}_{1}}\dots\sigma_{\hat{v}_{n}})}=\frac{\sum_{p,q}\E_{\T_{\delta}}(\sigma_{v_{1}}\dots\sigma_{v_{n}}\mu^{(p,q)})}{\sum_{p,q}\E_{\T_{\delta}}(\sigma_{\hat{v}_{1}}\dots\sigma_{\hat{v}_{n}}\mu^{(p,q)})}\\
=\frac{\E_{\T_{\delta}}(\sigma_{v_{1}}\dots\sigma_{v_{n}}\mu^{(0,0)})}{\E_{\T_{\delta}}(\sigma_{\hat{v}_{1}}\dots\sigma_{\hat{v}_{n}}\mu^{(0,0)})}\cdot\frac{\sum_{(p,q)}\frac{\E_{\T_{\delta}}(\sigma_{v_{1}}\dots\sigma_{v_{n}}\mu^{(p,q)})}{\E_{\T_{\delta}}(\sigma_{v_{1}}\dots\sigma_{v_{n}}\mu^{(0,0)})}}{\sum_{(p,q)}\frac{\E_{\T_{\delta}}(\sigma_{\hat{v}_{1}}\dots\sigma_{\hat{v}_{n}}\mu^{(p,q)})}{\E_{\T_{\delta}}(\sigma_{\hat{v}_{1}}\dots\sigma_{\hat{v}_{n}}\mu^{(0,0)})}}.
\end{multline*}
Note that at least for $\delta$ small enough, all these expectations
are positive: indeed, in view of (\ref{eq: conv_ratio_vs}), as we
vary $\hat{v}_{1},$ the ratio $\frac{\E_{\T_{\delta}}(\sigma_{\hat{v}_{1}}\sigma_{v_{2}}\dots\sigma_{v_{n}}\mu^{(p,q)})}{\E_{\T_{\delta}}(\sigma_{v_{1}}\dots\sigma_{v_{n}}\mu^{(p,q)})}$
stays positive, and after summing over $\hat{v}_{1}=(v_{1})_{kl}$
the numerator becomes $\E_{\T_{\delta}}(\sigma_{v_{1}}\sigma_{v_{2}}\dots\sigma_{v_{n}})$
which is also positive. Therefore, 
\begin{equation}
\frac{\E_{\T_{\delta}}(\sigma_{v_{1}}\dots\sigma_{v_{n}})}{\E_{\T_{\delta}}(\sigma_{\hat{v}_{1}}\dots\sigma_{\hat{v}_{n}})}=\exp\left(R_{\hat{v}_{1},\dots,\hat{v}_{n}}(v_{1},\dots,v_{n})+o(1)\right),\label{eq: conv_ratios_without_mu}
\end{equation}
 uniformly as in Corollary \ref{cor: conv_ratio}, where the smooth
function $R_{\hat{v}_{1},\dots,\hat{v}_{n}}$ is defined by
\[
\exp R_{\hat{v}_{1},\dots,\hat{v}_{n}}(v_{1},\dots,v_{n})=\exp R_{\hat{v}_{1},\dots,\hat{v}_{n}}^{(0,0)}(v_{1},\dots,v_{n})\frac{\sum_{(p,q)}\exp R_{v_{1},\dots,v_{n}}^{(0,0)}((v_{1})_{pq},\dots,v_{n})}{\sum_{(p,q)}\exp R_{\hat{v}_{1},\dots,\hat{v}_{n}}^{(0,0)}((\hat{v}_{1})_{pq},\dots,\hat{v}_{n})}.
\]

Let $D_{1},\dots,D_{n/2}\subset\T^{\delta}$ be disjoint discs with
centers $u_{1},\dots,u_{n/2}$ of small fixed radius $r.$ Using Edwards--Sokal
coupling and the FKG inequality, if $v_{2i-1},v_{2i}\in D_{i}$ for
each $i$, we can write 
\begin{multline}
\E_{D_{1},\mathrm{free}}(\sigma_{v_{1}}\sigma_{v_{2}})\dots\E_{D_{n/2},\mathrm{free}}(\sigma_{v_{n-1}}\sigma_{v_{n}})\\
\leq\E_{\T^{\delta}}(\sigma_{v_{1}}\dots\sigma_{v_{n}})\leq\E_{D_{1},+}(\sigma_{v_{1}}\sigma_{v_{2}})\dots\E_{D_{n/2},+}(\sigma_{v_{n-1}}\sigma_{v_{n}}),\label{eq: comparison_bound}
\end{multline}
 where $+$ and $\mathrm{free}$ stand for boundary conditions. 

We now claim that \eqref{eq: conv_ratios_without_mu} and \eqref{eq: comparison_bound}
together imply the result; for this argument, we will distinguish
points $v_{i}\in\T$ and sequences of points $v_{i}^{\delta}$ approximating
them. Let, as $\delta\to0$, $\hat{v}_{1}^{\delta},\dots,\hat{v}_{n}^{\delta}$
converge to distinct points $\hat{v}_{1},\dots,\hat{v}_{n}$ in $\T.$
By choosing $v_{2i-1}^{\delta},v_{2i}^{\delta}\in B_{u_{i}}(r/2),$
with $|v_{2i-1}^{\delta}-v_{2i}^{\delta}|>\frac{r}{4}$ and applying
(\eqref{eq: comparison_bound}) and the convergence of spin correlations
in simply-connected domains, we see that the expression
\[
C_{\sigma}^{-n}\delta^{-\frac{n}{8}}\E_{\T^{\delta}}(\sigma_{\hat{v}_{1}^{\delta}}\dots\sigma_{\hat{v}_{n}^{\delta}})=\exp\left(-R_{\hat{v}_{1},\dots,\hat{v}_{n}}(v_{1},\dots,v_{n})+o(1)\right)C_{\sigma}^{-n}\delta^{-\frac{n}{8}}\E_{\T^{\delta}}(\sigma_{v_{1}^{\delta}}\dots\sigma_{v_{n}^{\delta}}).
\]
 is bounded above and below as $\delta\to0.$ By passing to a subsequence,
we may assume that it converges to a limit, which we denote by $\ccor{\sigma_{\hat{v}_{1}},\dots,\sigma_{\hat{v}_{n}}}_{\T}.$
Now, for all $v_{2i-1}^{\delta},v_{2i}^{\delta}\in B_{u_{i}}(r/2)$
converging to distinct points, we have along that subsequence
\begin{multline*}
\exp\left(R_{\hat{v}_{1},\dots,\hat{v}_{n}}(v_{1},\dots,v_{n})\right)=\frac{C_{\sigma}^{-n}\delta^{-\frac{n}{8}}\E_{\T_{\delta}}(\sigma_{v_{1}^{\delta}}\dots\sigma_{v_{n}^{\delta}})}{C_{\sigma}^{-n}\delta^{-\frac{n}{8}}\E_{\T_{\delta}}(\sigma_{\hat{v}_{1}^{\delta}}\dots\sigma_{\hat{v}_{n}^{\delta}})}+o(1)\\
=\ccor{\sigma_{\hat{v}_{1}}\dots\sigma_{\hat{v}_{n}}}_{\T}^{-1}C_{\sigma}^{-n}\delta^{-\frac{n}{8}}\E_{\T_{\delta}}(\sigma_{v_{1}^{\delta}}\dots\sigma_{v_{n}^{\delta}})+o(1).
\end{multline*}
We divide both sides by $\prod_{i=1}^{n/2}|v_{2i-1}-v_{2i}|^{-\frac{1}{4}}$
and observe that in view of \eqref{eq: comparison_bound} and the
properties of scaling limits of correlations in the disc, we have
\[
\lim_{v_{2i-1},v_{2i}\to u_{i}}\lim_{\delta\to0}\frac{C_{\sigma}^{-n}\delta^{-\frac{n}{8}}\E_{\T_{\delta}}(\sigma_{v_{1}^{\delta}}\dots\sigma_{v_{n}^{\delta}})}{\prod_{i=1}^{n/2}|v_{2i-1}-v_{2i}|^{-\frac{1}{4}}}=1.
\]
 Therefore, combining the last two displays, we deduce 
\[
\ccor{\sigma_{\hat{v}_{1}}\dots\sigma_{\hat{v}_{n}}}_{\T}=\lim_{v_{2i-1},v_{2i}\to u_{i}}\frac{\exp\left(-R_{\hat{v}_{1},\dots,\hat{v}_{n}}(v_{1},\dots,v_{n})\right)}{\prod_{i=1}^{n/2}|v_{2i-1}-v_{2i}|^{\frac{1}{4}}}.
\]
Because the right-hand side does not depend on the approximating sequence
$\hat{v}_{i}^{\delta},$ nor on the subsequence, we infer that it
is in fact simply the limit of $C_{\sigma}^{-n}\delta^{-\frac{n}{8}}\E_{\T^{\delta}}(\sigma_{\hat{v}_{1}^{\delta}}\dots\sigma_{\hat{v}_{n}^{\delta}}).$
It is clear from the construction that 
\[
\frac{\ccor{\sigma_{v_{1}}\dots\sigma_{v_{n}}}_{\T}}{\ccor{\sigma_{\hat{v}_{1}}\dots\sigma_{\hat{v}_{n}}}_{\T}}=\exp\left(R_{\hat{v}_{1},\dots,\hat{v}_{n}}(v_{1},\dots,v_{n})\right)
\]
 and that $\ccor{\sigma_{v_{1}},\dots,\sigma_{v_{n}}}_{\T}$ satisfies
(\ref{eq: merge_asymptotics}). Finally, we can write 
\begin{multline*}
C_{\sigma}^{-n}\delta^{-\frac{n}{8}}\E_{\T^{\delta}}(\sigma_{v_{1}}\dots\sigma_{v_{n}}\mu^{(p,q)})=C_{\sigma}^{-n}\delta^{-\frac{n}{8}}\E_{\T^{\delta}}(\sigma_{v_{1}}\dots\sigma_{v_{n}})\cdot\frac{\E_{\T^{\delta}}(\sigma_{v_{1}}\dots\sigma_{v_{n}}\mu^{(p,q)})}{\E_{\T^{\delta}}(\sigma_{v_{1}}\dots\sigma_{v_{n}})}\\
\stackrel{\delta\to0}{\longrightarrow}\ccor{\sigma_{v_{1}},\dots,\sigma_{v_{n}}\mu^{(p,q)}}_{\T}.
\end{multline*}
\end{proof}

\subsection{\label{subsec: all_correlations}Spin-fermion and spin-energy correlations.}

As in the planar case, the definition of fermionic observables $F_{v_{1},\dots,\v_{n}}(w,z)$
and $F_{v_{1},\dots,\v_{n}}^{(p,q)}(w,z)$ can be generalized to the
case of multiple points. These observables, and the convergence results
for them, can then be reduced to $F_{v_{1},\dots,\v_{n}}(w,z)$ and
$F_{v_{1},\dots,\v_{n}}^{(p,q)}(w,z)$ via \emph{Pfaffian} \emph{identities}.
Since there is essentially no difference with the planar case \cite{CHI2},
or the torus case with $n=0$ and $(p,q)\neq(0,0)$, our treatment
will be brief. The multi-point fermionic observable is defined as
\[
F_{v_{1},\dots,\v_{n}}^{(p,q)}(z_{1},\dots,z_{2k})=\Corr{\Td}{\psi_{z_{1}}\dots\psi_{z_{2k}}\svv\mu^{(p,q)}},
\]
 where $\psi_{z_{i}}$ are as in Definition \ref{def: observable}.
As a function, say, of $z_{1},$ it is well-defined on $\Cgr_{\{z_{2},\dots,z_{n}\}}(\Tdouble_{[v_{1},\dots,v_{n}]})$
up to a global sign choice and satisfies the spinor and anti-periodicity
properties (\ref{eq: F_spinor}), (\ref{eq: F_antiperiodic}); the
same proof as in Lemma \ref{lem: properties of F} applies verbatim.
We can fix the global sign by a recursive convention as in \cite[Definition 2.23]{CHI2}.
Defining the normalized observables as 
\[
\hat{F}_{v_{1},\dots,\v_{n}}^{(p,q)}(z_{1},\dots,z_{2k})=\frac{\Corr{\Td}{\psi_{z_{1}}\dots\psi_{z_{2k}}\svv\mu^{(p,q)}}}{\Corr{\Td}{\svv\mu^{(p,q)}}},
\]
 we have the following identity, unless $n=p=q=0$:
\begin{equation}
\hat{F}_{v_{1},\dots,\v_{n}}^{(p,q)}(z_{1},\dots,z_{2k})=\Pf\left[\hat{F}_{v_{1},\dots,\v_{n}}^{(p,q)}(z_{i},z_{j})\right].\label{eq: Pfaffian_identity}
\end{equation}
For the proof of (\ref{eq: Pfaffian_identity}), notice that the proof
of \cite[Proposition 2.24]{CHI2} only relies on the s-holomorphicity
and the uniqueness in the corresponding boundary value problem. Since
we have the corresponding uniqueness result in the torus case in Lemma
\ref{lem: no_kernel}, exactly the same proof applies. Alternatively,
a combinatorial proof (cf. \cite[Section 4 ]{CCK}, especially the
discussion at the end) applies, since the only possible obstacle would
be the vanishing of the determinant, which does not happen by Lemma
\ref{lem: denom_does_not_vanish}.

For an edge $e=(e_{+}e_{-})$ of $\T^{\delta},$ the \emph{energy
observable }is defined by\emph{ 
\[
\en_{e}=\sqrt{2}\sigma_{e_{+}}\sigma_{e_{-}}-1.
\]
}As explained e.g. in \cite[Section 2.5]{CHI2}, this observable can
be written as $\mu_{e^{\star}}\sigma_{e_{+}}\sigma_{e_{-}}=\eta_{z}^{-1}\eta_{\hat{z}}^{-1}\psi_{z}\psi_{\hat{z}},$
where $e^{\star}$ is the edge dual to $e$, and $z,\hat{z}$ are
two corners incident to $e$ and but not to the same vertex. Thus,
the spin-energy correlations also obeys a Pfaffian formula 
\begin{multline}
\frac{\Corr{\Td}{\en_{e_{1}}\dots\en_{e_{l}}\svv\mu^{(p,q)}}}{\Corr{\Td}{\svv\mu^{(p,q)}}}=\eta_{z_{1}}^{-1}\eta_{\hat{z}_{1}}^{-1}\dots\eta_{z_{l}}^{-1}\eta_{\hat{z}_{l}}^{-1}\frac{\E\left(\psi_{z_{1}}\psi_{\hat{z}_{1}}\dots\psi_{z_{l}}\psi_{\hat{z}_{l}}\svv\mu^{(p,q)}\right)}{\Corr{\Td}{\svv\mu^{(p,q)}}}\\
=\Pf\left[\frac{\left(\eta_{z}\eta_{w}\right)^{-1}\E\left(\psi_{z}\psi_{w}\svv\mu^{(p,q)}\right)}{\Corr{\Td}{\svv\mu^{(p,q)}}}\right]_{z,w\in\{z_{1},\hat{z}_{1},\dots,z_{l},\hat{z}_{l}\}}.\label{eq: Pfaff_energy}
\end{multline}

\begin{cor}
\label{cor: all_corr}Let $n>0$ be even. For every $\eps>0$, we
have the following asymptotics as $\delta\to0$, uniformly as as long
as the marked points stay at the distance at least $\eps$ from each
other:

\[
\]
\begin{equation}
\delta^{-\frac{k}{2}}C_{\psi}^{-2k}\frac{\Corr{\Td}{\psi_{z_{1}}\dots\psi_{z_{2k}}\svv\mu^{(p,q)}}}{\Corr{\Td}{\svv\mu^{(p,q)}}}=\eta_{z_{1}}\dots\eta_{z_{2k}}\Pf[A_{ij}]_{i,j=1}^{2k}+o(1),\label{eq: conv_psi_spins}
\end{equation}
\begin{equation}
\delta^{-l}C_{\eps}^{-l}\frac{\Corr{\Td}{\en_{e_{1}}\dots\en_{e_{l}}\svv\mu^{(p,q)}}}{\Corr{\Td}{\svv\mu^{(p,q)}}}=\left(-\frac{\i}{2}\right)^{l}\Pf[B_{ij}]_{i,j=1}^{2l}+o(1),\label{eq: conv_en_spins}
\end{equation}
where $C_{\psi}=\left(\frac{2}{\pi}\right)^{\frac{1}{2}},$ $C_{\eps}=\frac{2}{\pi},$
and $A_{ii}=B_{ii}=0$ for all $i,$ while for $i\neq j,$ we put
$A_{ij}=\re[\bar{\eta}_{z_{i}}f_{v_{1},\dots,v_{n}}^{[\eta_{z_{j}}],(p,q)}(z_{i},z_{j})]$
and 
\[
B_{ij}=\begin{cases}
f_{v_{1},\dots,v_{n}}^{(p,q)}(e_{\left\lceil i/2\right\rceil },e_{\left\lceil j/2\right\rceil }), & i=1\mod2,\quad j=1\mod2;\\
\overline{f_{v_{1},\dots,v_{n}}^{(p,q)}(e_{\left\lceil i/2\right\rceil },e_{\left\lceil j/2\right\rceil })}, & i=2\mod2,\quad j=2\mod2;\\
f_{v_{1},\dots,v_{n}}^{\star,(p,q)}(e_{\left\lceil i/2\right\rceil },e_{\left\lceil j/2\right\rceil }), & i=2\mod2,\quad j=1\mod2;\\
\overline{f_{v_{1},\dots,v_{n}}^{\star,(p,q)}(e_{\left\lceil i/2\right\rceil },e_{\left\lceil j/2\right\rceil })}, & i=1\mod2,\quad j=2\mod2.
\end{cases}
\]
\end{cor}

\begin{proof}
The proof of the first statement is a combination of (\ref{eq: Pfaffian_identity})
and Theorem \ref{thm: convergence of the observable}. The proof of
the second one is a combination of (\ref{eq: Pfaff_energy}), Theorem
\ref{thm: convergence of the observable}, a local analysis of the
singular terms $\frac{\left(\eta_{z}\eta_{w}\right)^{-1}\E\left(\psi_{z_{i}}\psi_{\hat{z}_{i}}\svv\mu^{(p,q)}\right)}{\Corr{\Td}{\svv\mu^{(p,q)}}}$
in the Pfaffian (see \cite[Theorem 3.21]{CHI2}), and a linear change
of variables in the Pfaffian, as in \cite[Proof of Theorem 1.2]{CHI2}.
\end{proof}
Recall that the above corollary holds in the same form for $n=0$
and $(p,q)\neq0$; the case $n=0$ and $(p,q)=0$ requires a separate
treatment \cite{IKT}. Plainly, combining Corollary \ref{cor: all_corr}
with Theorem \ref{thm: conv_spins}, one gets the asymptotics of the
numerators in (\ref{eq: conv_psi_spins}) and (\ref{eq: conv_en_spins}).
We also mention that one can obtain in a similar way the mixed energy-spin-fermion
correlations, as well as disorder correlations, by placing some of
$\psi_{z_{i}}$ next to one of $v_{j},$ and applying the proof of
\cite[Theorem 3.32]{CHI2}, which is a ``local'' consequence of
\cite[Theorem 3.16]{CHI2}. We leave details to the reader.

\section{\label{sec: Szego}Szegö kernels and their limits on Riemann surfaces. }

In this section, we construct holomorphic functions, or, more precisely,
half-integer differentials, satisfying (\ref{eq: contin_period},
\ref{eq: contin_pole}, \ref{eq: contin_sing}) by a ``continuous''
construction, and give explicit formulae for them in terms of Abelian
differentials and theta functions. Since this part does not depend
on the particularities of the torus case, we carry it out in general,
and specialize to the case of a torus in the next subsection. The
construction is based on gluing to a Riemann surface $\M$ its copy
$\Mtilde$ by thin ``tubes'' attached near the marked points $v_{1},\dots,v_{n}$
and considering the \emph{Szegő kernel} on the resulting surface;
the functions $\ffs{w,\cdot}{\vv}{(p,q)}$ and $\ffs{w,\cdot}{\vv}{\star,(p,q)}$
are then recovered by ``pinching'' the tubes. The analysis here
is similar to the one carried out in \cite[Section 4]{BIVW}, with
a difference that in the absence of the boundary, removing the last
tube disconnects the surfaces. This leads to a difference in the resulting
formulae; in particular, to the well-known ``charge neutrality''
condition.

Our starting point is the following identity due to Hejhal \cite{Hejhal}
and Fay \cite{Fay}:
\begin{equation}
\Lambda_{\M,\SpStr}(P,Q)^{2}=\beta_{\M}(P,Q)+\sum_{i,j=1}^{g}\frac{\partial_{z_{i}}\partial_{z_{j}}\theta_{\tau_{\M}}(0;H)}{\theta_{\tau_{\M}}(0;H)}u_{\M,i}(P)u_{\M,j}(Q).\label{eq: bosonization}
\end{equation}
Here $\M$ is a compact Riemann surface of genus $g$, $\Lambda_{\M,\SpStr}(P,Q)$
is the Szegő kernel on $\M$ corresponding to the spin structure $\SpStr$,
$\tau_{\M}$ is the period matrix of $\M$, $\theta_{\tau_{\M}}(\cdot,H)$
is the theta function with half-integer characteristic $H$ corresponding
(cf. \cite[Lemma 3.14]{BIVW}) to $\SpStr$, $\beta_{\M}$ is the
fundamental Abelian differential of the second kind on $\M,$ and
$u_{\M,i}(w),$ $i=1,\dots,g$, are the Abelian differentials of the
first kind. We will assume that the reader is familiar with all these
objects, and refer to \cite[Section 3]{BIVW} for notation, detailed
background and explanations; we note our normalization conventions
\[
\int_{A_{i}}u_{\M,j}=\pi\i\delta_{ij},\quad\tau_{ij}=\int_{B_{i}}u_{\M,j}.
\]
 The Szegő kernel $\Lambda_{\M,\SpStr}$ exists and is unique, and
the identity is valid, as long as $\theta_{\tau_{\M}}(0;H)\neq0$
\cite{Hejhal}.

\subsection{The surface $\protect\M_{\protect\eps}$ and the limits of Abelian
differentials and the period matrix.}

Given a compact Riemann surface $\M$ of genus $g$ with distinct
marked points $\vv\in\M$, and $\eps=(\eps_{1},\dots,\eps_{n})\in(\C\setminus\{0\})^{n},$
we construct a new surface $\Meps$ of genus $2g+n-1$, as follows.
For a coordinate map $z:\Omega\to\C,$ where $\Omega\subset\M,$ we
denote $z^{\star}(P):=\overline{z(P)},$ and let $(\Mtilde,v_{1}^{\star},\dots,v_{n}^{\star})$
be an identical copy of $(\M,v_{1},\dots,v_{n})$ endowed with the
analytic atlas given by the maps $z^{\star}.$ Following \cite{Yamada},
we take $z_{i}:\Omega_{i}\to2\D,$ $z_{i}^{\star}:\Omega_{i}^{\star}\to2\D$
to be local coordinates in the (pairwise disjoint, simply connected)
neighborhoods $\Omega_{i}\subset\M,$ $\Omega_{i}^{\star}\subset\Mtilde$
of $v_{i},v_{i}^{\star}$ respectively, such that $z_{i}(v_{i})=z_{i}^{\star}(v_{i}^{\star})=0$
(in our case, we also take them to be copies of the same map) and
identify points in annuli $\{P\in\M:|\eps_{i}|\leq|z_{i}(P)|\leq1)\}$
and $\{P\in\Mtilde:|\eps_{i}|\leq|z_{i}^{\star}(P)|\leq1)\}$ by the
equation $\{z_{i}z_{i}^{\star}=\eps_{i}\}.$ Note that this procedure
is equivalent to cutting out discs $\{P\in\M:|z_{i}(P)|\leq|\eps_{i}|^{\frac{1}{2}})\}$
and $\{P\in\Mtilde:|z_{i}^{\star}(P)|\leq|\eps_{i}|^{\frac{1}{2}})\}$
and gluing surfaces along their boundaries, cf. \cite[Section 4.1]{BIVW}.
For our purposes, it will be enough to consider positive $\eps_{i}$
only, although allowing them to be complex is natural since the quantities
of interest will be holomorphic in these parameters. The coordinates
$z_{i},z_{i}^{\star}$ are called \emph{pinching coordinates}, and
the annuli $\mathcal{A}_{i}(\eps_{i}):=\{P\in\M:|\eps_{i}|\leq|z_{i}(P)|\leq1)\}$
the \emph{pinching annuli}. We introduce auxiliary points $\hat{v}_{i}=z_{i}^{-1}(1),\hat{v}_{i}^{\star}=(z_{i}^{\star})^{-1}(1)$
in these annuli.
\begin{defn}
\label{def:markchar} 
We will choose the following Torelli marking (i.e., a base of homology)
for $\Meps$:
\begin{enumerate}
\item the loops $A_{i},B_{i}$, $i=1,2,\dots,g$ (respectively, $i=g+1,g+2,\dots,2g$)
are the cycles in a Torelli marking of $\M$ (respectively, their
copies in $\Mtilde$, with orientations reversed for $B_{i}$). 
\item the loop $A_{2g+i-1}$, $i=2,\dots,n$, is a small simple loop in
$\M$ surrounding $v_{i}$ counterclockwise. The loop $B_{2g+i-1}$,
$i=2,\dots,n$ goes on $\M$ from $\hat{v}_{1}$ to $\hat{v}_{i}$,
then continues through $\Ann_{i}(\eps_{i})$ to $\hat{v}_{i}^{\star}$,
then takes a symmetric path on $\Mtilde$ to $\hat{v}_{1}^{\star},$
and then goes through $\Ann_{1}(\eps_{1})$ to $\hat{v}_{1}$. We
insist that $B_{2g+i-1}$ avoids all other loops except $A_{2g+i-1}$;
when $\eps_{i}$ are real, assume for definiteness the paths through
the pinching regions to be straight line segments.
\end{enumerate}
Given a spin structure $\SpStr$ on $\M$, or equivalently a half-integer
characteristics $H\in\{0,1\}^{2g}$, $H=\{M_{i},N_{i}\}_{i=1}^{g}$
with $\SpStr=\SpStr(H)$ (see \cite[Lemma 3.14]{BIVW}), we introduce
a half-integer theta characteristic $H\in\{0,1\}^{2\cdot(2g+n-1)}$,
by extending the $2g$-tuple $\{M_{i},N_{i}\}_{i=1}^{g}$ as follows: 
\begin{enumerate}
\item for $i=1,2,\dots g$, we set $M_{i+g}=M_{i}$, $N_{i+g}=N_{i}$. 
\item for $i=2g+1,\dots,2g+n-1$, we set $N_{i}=1$ and $M_{i}=0.$ 
\end{enumerate}
\end{defn}

Note that we need to single out one of the marked points, $v_{1}$,
for this construction, but we will see that the final result will
be independent of this choice.

Similar to \cite{BIVW}, we will need to take limits of Abelian differentials
and the period matrix as $\eps_{i}\rightarrow0$ to obtain the Szegő
kernel on $\Mhat=\M\sqcup\Mtilde$.
\begin{defn}
\label{def: u_j_all} We introduce the following unifying notation
for the Abelian differentials of the \emph{first and third kind} on
$\Mhat$: 
\begin{align}
u_{i} & :=u_{\M_{0},i}, & i=1,\dots,2g\label{eq: u_i_all_1}\\
u_{2g+i-1} & :=\frac{1}{2}\omega_{\M_{0},v_{i},v_{1}}+\frac{1}{2}\omega_{\M_{0},v_{1}^{\star},v_{i}^{\star}} & i=2,\dots,n.\label{eq: u_i_all_2}
\end{align}
\end{defn}

Hereinafter we denote by $\omega_{\M,u_{1},u_{2}}$ (the unique) Abelian
differential of the third kind on $\M$ with simple poles of residue
$1$ at $u_{1}$ and of residue $-1$ at $u_{2}$, and vanishing $A_{i}$-periods,
or zero if it does not exist (which happens if and only if $u_{1},u_{2}$
belong to different connected components of $\M$). Note that for
$1\leq i\leq g$ (respectively, $g+1\leq i\leq2g$), we have $u_{i}=u_{\M,i}$
on $\M$ and $u_{i}\equiv0$ on $\Mtilde$ (respectively, $u_{i}\equiv0$
on $\M$ and $u_{i}=u_{\M^{\star},i-g}$ on $\Mtilde$), and for $i=2,\dots,n$,
we have $u_{2g+i-1}=\frac{1}{2}\omega_{\M,v_{i},v_{1}}$ on $\M$
and $u_{2g+i-1}=\frac{1}{2}\omega_{\M^{\star},v_{1}^{\star},v_{i}^{\star}}$
on $\Mtilde.$ Similarly, for the fundamental Abelian differential
of the second kind, we have 
\[
\beta_{\Mhat}(P,Q)=\begin{cases}
\beta_{\M}(P,Q), & P,Q\in\M\\
\beta_{\Mtilde}(P,Q) & P,Q\in\Mtilde\\
0, & \text{else}.
\end{cases}
\]

We will need the following notion of \emph{reguralized integrals}
for $u_{i}:$ for $\omega$ an Abelian differential with at worst
simple poles, we put 
\[
\regint{\int_{u_{1}}^{u_{2}}\omega}=\lim_{x\to0}\left(\int_{w_{1}^{-1}(x)}^{w_{2}^{-1}(x)}\omega-(\res_{u_{1}}w)\log(x)-(\res_{u_{1}}w)\log(x)\right),
\]
where $w_{1},w_{2}$ are local coordinates such that $w_{i}(u_{i})=0.$
The above notion is independent of $w_{i}$ if $\omega$ is holomorphic
at $u_{i};$ otherwise, it depends on the local coordinates, the path
taken between $u_{1}$ and $u_{2}$, and the branch of the logarithm.
We will only use it when $\{u_{1},u_{2}\}=\{v_{1},v_{i}\}$ or $\{u_{1},u_{2}\}=\{v_{1}^{\star},v_{i}^{\star}\}$
for $i=2,\dots,n,$ in which case we choose $w_{i}$ to be the pinching
coordinates, the path of integration follows a piece of $B_{2g+i-1}$
between $\hat{v}_{1},\hat{v}_{i}$ or $\hat{v}_{1}^{\star},\hat{v}_{i}^{\star}$
appended with $w_{1}^{-1}([0,1])$ and $w_{2}^{-1}([0,1]),$ and the
logarithm to be the principal branch.
\begin{lem}
\label{lem:abdifasy} \cite[Cor. 1, Cor. 3, Cor. 5 and Theorem 2, Theorem 6]{Yamada}
The Abelian differentials on $\Meps$ have the following limits as
$\eps_{i}\rightarrow0$: \label{lem: limit_Abelian} 
\begin{align}
u_{\Meps,j} & =u_{j}+O(\eps),\quad j=1,\dots,2g+n-1,\label{eq: conv_first}\\
\beta_{\Meps} & =\beta_{\Mhat}+O(\eps).\label{eq: conv_second}
\end{align}
locally uniformly on $\M_{0}\setminus\{v_{1},v_{1}^{\star},\dots,v_{n},v_{n}^{\star}\}.$ 
\end{lem}

\begin{proof}
The proofs in \cite{Yamada} are given for the case of pinching one
handle, however the argument proceeds by showing holomorphicity in
each parameter $\eps_{i}.$ Hartogs' theorem on separate holomorphicity
then implies the general result.
\end{proof}
We also have an expansion for the period matrix:
\begin{lem}
\label{lem: limit_period} As $\eps=(\eps_{1},\dots,\eps_{n})\rightarrow0$,
the period matrix $\tau_{\eps}=\tau_{\Meps}$ of $\Meps$ has the
following expansion:

\begin{equation}
\tau_{\eps}=\logmatrix_{\eps}+\tau_{0}+O(\eps),
\end{equation}
where 
\[
\logmatrix_{\eps}=\frac{1}{2}\mathbf{0}_{2g\times2g}\oplus\left(\begin{array}{cccc}
\log(\eps_{1}\eps_{2}) & \log\eps_{1} & \dots & \log\eps_{1}\\
\log\eps_{1} & \log(\eps_{1}\eps_{3}) &  & \vdots\\
\vdots &  & \ddots & \log\eps_{1}\\
\log\eps_{1} & \dots & \log\eps_{1} & \log(\eps_{1}\eps_{n})
\end{array}\right)
\]
 and the symmetric matrix $\tau_{0}$ is given by 
\begin{align*}
(\tau_{0})_{ij}=\begin{cases}
(\tau_{\M_{0}})_{ij} & 1\leq i,j\leq2g;\\
\regint{\int_{v_{1}}^{v_{j+1-2g}}u_{i}}+\regint{\int_{v_{j+1-2g}^{\star}}^{v_{1}^{\star}}u_{i}}, & j>2g.
\end{cases}
\end{align*}
\end{lem}

\begin{proof}
This lemma is similar to \cite[Cor. 5, Lem. 5]{Yamada}, where the
case of one pinched handle is done, and can probably be reduced thereto
using Hartogs' separate holomorphicity theorem. We prefer to give
a self-contained proof which is simpler than one in \cite{Yamada}
since we only care about the expansion up to a constant term. 

The $1\leq i,j\leq2g$ part is clear, since in the case, $u_{\M_{\eps},i}\to u_{\M_{0},i}$
uniformly on $B_{j},$ so we proceed by considering $i,j>2g.$ In
this case, we need to compute the asymptotics of the integral $\int_{B_{j}}u_{\M_{\eps},i},$
where $B_{j}$ goes through the pinching annuli at $v_{1}$ and $v_{k},$
where $k=j+1-2g.$ We write 
\begin{equation}
\int_{B_{j}}u_{\M_{\eps},i}=\left(\int_{\hat{v}_{1}}^{\hat{v}_{k}}+\int_{\hat{v}_{k}}^{\hat{v}_{k}^{\star}}+\int_{\hat{v}_{k}^{\star}}^{\hat{v}_{1}^{\star}}+\int_{\hat{v}_{1}^{\star}}^{\hat{v}_{1}}\right)u_{\M_{\eps},i}.\label{eq: decomp}
\end{equation}
The first and the third integral are well-behaved as $\eps_{i}\to0$
and simply converge to the integrals of $u_{i}$ over the corresponding
paths. We proceed to estimating the last integral, i.e., the one passing
through the pinching region. For $0<\rho<2,$ denote $C_{\rho}:=\{z_{1}^{-1}(\rho e^{i\theta}):0\leq\theta<2\pi\}$,
and recall that $\int_{C_{\rho}}u_{\M_{\eps},i}=-\pi\i$. Hence, for
any $w\in\mathcal{A}_{1}(\eps_{i}),$
\[
\Theta_{w}^{\eps}(z):=\int_{w}^{z}\left(u_{\M_{\eps},i}+\frac{dz_{1}}{2z_{1}}\right)=\int_{w}^{z}\left(u_{\M_{\eps},i}-\frac{dz_{1}^{\star}}{2z_{1}^{\star}}\right),
\]
is a \emph{single-valued} holomorphic function. Taking into account
that $\int_{\hat{v}_{1}^{\star}}^{\hat{v}_{1}}\left(u_{\M_{\eps},i}+\frac{dz_{1}}{2z_{1}}\right)=\Theta_{\hat{v}_{1}^{\star}}^{\eps}(z)-\Theta_{\hat{v}_{1}}^{\eps}(z)$
for all $z\in\mathcal{A}_{1}(\eps_{i}),$ we can write, for any $\rho\in(\eps_{i},1)$
\begin{multline*}
\int_{\hat{v}_{1}^{\star}}^{\hat{v}_{1}}\left(u_{\M_{\eps},i}+\frac{dz_{1}}{2z_{1}}\right)=-\frac{1}{\pi\i}\int_{C_{\rho}}\int_{\hat{v}_{1}^{\star}}^{\hat{v}_{1}}\left(u_{\M_{\eps},i}+\frac{dz_{1}}{2z_{1}}\right)u_{\M_{\eps},i}\\
=\frac{1}{\pi\i}\left(\int_{C_{\rho}}\Theta_{\hat{v}_{1}}^{\eps}u_{\M_{\eps},i}-\int_{C_{\rho}}\Theta_{\hat{v}_{1}^{\star}}^{\eps}u_{\M_{\eps},i}\right).
\end{multline*}
We now choose $\rho_{1}=|z_{1}(v_{k})|=1$ in the first integral,
and $\rho_{2}=|z_{1}(v_{k}^{\star})|=\eps_{k}$ in the second one,
so that $C_{\rho_{2}}=-C_{1}^{\star}=\{\left(z_{1}^{\star}\right)^{-1}(\rho e^{-i\theta}):0\leq\theta<2\pi\},$
and note that now, the integrands converge on $C_{1}$ and $C_{1}^{\star}$
respectively. Therefore, 
\[
\int_{\hat{v}_{1}^{\star}}^{\hat{v}_{1}}u_{\M_{\eps},i}=\frac{1}{2}\log(\eps_{1})+\frac{1}{\pi\i}\left(\int_{C_{1}}\Theta_{\hat{v}_{1}}^{0}u_{i}+\int_{C_{1}^{\star}}\Theta_{\hat{v}_{1}^{\star}}^{0}u_{i}\right)+O(\eps),
\]
where $\Theta_{w}^{0}(z)=\int_{w}^{z}\left(u_{i}+\frac{dz_{1}}{2z_{1}}\right)$
and $u_{i}=\frac{1}{2}\omega_{\M_{0},v_{k},v_{1}}+\frac{1}{2}\omega_{\M_{0},v_{1}^{\star},v_{k}^{\star}}.$
Now, recalling that $u_{i}$ is meromorphic on $\Omega_{1}\cup\Omega_{1}^{\star}$
with a simple pole at $v_{1},v_{1}^{\star}$ of residue $\mp\frac{1}{2},$
the above expression evaluates to
\begin{multline*}
\int_{\hat{v}_{1}^{\star}}^{\hat{v}_{1}}u_{\M_{\eps},i}=\frac{1}{2}\log\eps_{1}-\Theta_{\hat{v}_{1}}^{0}(v_{1})+\Theta_{\hat{v}_{1}^{\star}}^{0}(v_{1}^{\star})+O(\eps)\\
=\frac{1}{2}\log\eps_{1}-\lim_{x\to0}\left(\int_{\hat{v}_{1}}^{z_{1}^{-1}(x)}u_{i}+\frac{1}{2}\log x\right)+\lim_{x\to0}\left(\int_{\hat{v}_{1}^{\star}}^{(z_{1}^{\star})^{-1}(x)}u_{i}-\frac{1}{2}\log x\right)+O(\eps).
\end{multline*}
If $i=j,$ we treat the second term in (\ref{eq: decomp}) in exactly
the same way, leading to 
\[
\int_{\hat{v}_{k}}^{\hat{v}_{k}^{\star}}u_{\M_{\eps},i}=\frac{1}{2}\log\eps_{k}+\lim_{x\to0}\left(\int_{\hat{v}_{k}}^{z_{k}^{-1}(x)}u_{i}-\frac{1}{2}\log x\right)-\lim_{x\to0}\left(\int_{\hat{v}_{k}^{\star}}^{(z_{k}^{\star})^{-1}(x)}u_{i}+\frac{1}{2}\log x\right)+O(\eps).
\]
If $i\neq j,$ the function $\int_{w}^{z}u_{\M_{\eps},i}$ is itself
single-valued on $\Ann_{k}(\eps_{k}),$ so plugging it the place of
$\Theta_{w}^{\eps}(z)$ in the above argument and taking into account
that $u_{i}$ does not have a pole at $v_{k},v_{k}^{\star}$ leads
to 
\[
\int_{\hat{v}_{k}^{\star}}^{\hat{v}_{k}}u_{\M_{\eps},i}=\int_{\hat{v}_{k}^{\star}}^{v_{k}^{\star}}u_{i}+\int_{v_{k}}^{\hat{v}_{k}}u_{i}+O(\eps).
\]
 When $i\leq2g,$ a similar analysis applies to both pinching regions.
In all cases, putting together the four contributions to (\ref{eq: decomp})
gives 

\[
\int_{B_{j}}u_{\M_{\eps},i}=\frac{1}{2}\log\eps_{1}+\frac{1}{2}\delta_{ij}\log\eps_{j+1-2g}+\left(\int_{v_{1}}^{v_{j+1-2g}}u_{i}\right)_{\mathrm{reg}}+\left(\int_{v_{j+1-2g}^{\star}}^{v_{1}^{\star}}u_{i}\right)_{\mathrm{reg}}+O(\eps)
\]
\end{proof}
\begin{rem}
\label{rem: symmetries}The above Lemmas \ref{lem: limit_Abelian}
and \ref{lem: limit_period} do not use in any way that $\Mtilde$
is a copy of $\M$ with complex structure inverted; we now outline
how this symmetry affect the resulting expressions. Let $\inv:P\mapsto P^{\star}$
denotes the anti-holomorphic involution on $\M_{\eps},$ mapping a
point on $\M$ to the corresponding point on $\Mtilde$ and vice versa.
If $\omega$ is a holomorphic differential on $\M_{\eps},$ then its
pullback $\inv^{*}\omega$ is anti-holomorphic, hence $\overline{\inv^{*}\omega}$
is again holomorphic. This means, in particular, that $u_{i}=-\overline{\inv^{*}u_{i-g}}$
for $g+1\leq i\leq2g,$ the $-$ sign coming from conjugating $\i$
in the normalization of $u_{i}.$ Therefore, for $g+1\leq i,j\leq2g$,
we have $\tau_{ij}=\overline{\tau_{i-g,j-g}}$. Similarly, for $2g+1\leq i,j\leq2g+n-1,$
we have $u_{i}=-\overline{\inv^{*}u_{i}},$ and hence in this case,
$\tau_{ij}=2\re\left(\int_{v_{1}}^{v_{j+1-2g}}u_{i}\right)_{\mathrm{reg}}.$
Finally, using remark after Definition \ref{def: u_j_all}, we also
see that $\tau_{ij}=0$ for $1\leq i\leq g$ and $g+1\leq j\leq2g,$
and also we can write 
\[
\tau_{ij}=\begin{cases}
\left(\int_{v_{1}}^{v_{j+1-2g}}u_{i}\right), & 1\leq i\leq g,2g+1\leq j\leq2g+n-1,\\
\overline{\left(\int_{v_{1}}^{v_{j+1-2g}}u_{i-g}\right),} & g+1\leq i\leq2g,2g+1\leq j\leq2g+n-1.
\end{cases}
\]
\end{rem}

\subsection{The limiting form of (\ref{eq: bosonization}).}

In this section, we will obtain a limit of the right-hand side of
\eqref{eq: bosonization}. Recall from Definition \ref{def:markchar}
the vectors $M,N\in\{0,1\}^{2g+n-1}$ that fix the theta characteristics
$H$. For $1\leq p\leq2g+n-1$, we introduce the relation $s_{p}=2m_{p}+N_{p}$
between the parameters $m_{p}$ and $s_{p}$. We denote 
\begin{align*}
\setS & :=(2\Z+N_{1})\times\dots\times(2\Z+N_{2g+n-1})\\
 & =(2\Z+N_{1})^{2}\times(2\Z+N_{2})^{2}\times\dots\times(2\Z+N_{g})^{2}\times(2\Z+1)^{n-1}\\
\setSn & :=\prod_{i=1}^{g}(2\Z+N_{i})^{2}\times\{s_{1},\cdots,s_{n-1}|s_{j}\in\{\pm1\}^{n-1},\sum_{j=1}^{n-1}s_{j}\in\{\pm1\}\},
\end{align*}
so that as $m$ runs through $\Z^{2g+n-1}$, $s$ runs through $\setS$. 
\begin{lem}
\label{limitlemma}

Assuming that $\sum_{s\in\setSn}\exp(Q_{\tau_{0}}(s)+(\pi\i/2)(s\cdot M)))\neq0$,
as $\eps\to0$, the right-hand side of \eqref{eq: bosonization} converges,
locally uniformly, to

\begin{align}
\beta_{\Mhat}(P,Q)+\sum_{k,l=1}^{2g+n-1}\frac{\sum_{s\in\setSn}s_{k}s_{l}\exp(Q_{\tau_{0}}(s)+(\pi\i/2)(s\cdot M)))}{\sum_{s\in\setSn}\exp(Q_{\tau_{0}}(s)+(\pi\i/2)(s\cdot M)))}u_{k}(P)u_{l}(Q),\label{eq: RHS_limit}
\end{align}
where $u_{k}$ and $\tau_{0}$ are as given in Definition \ref{def: u_j_all}
and Lemma \ref{lem: limit_period}. 
\end{lem}

\begin{proof}
The proof is very similar to the proof of Lemma 4.9 in \cite{BIVW}.
We will shrink the handles at the same rate, re-denoting $\eps_{i}:=\eps$
for all $i$. Due to Lemma \ref{lem: limit_Abelian}, we have $\beta_{\Meps}\to\beta_{\Mhat}$
and $u_{\Meps,j}\to u_{j}$. Thus, we only have to work out the convergence
of the coefficients $\pa_{z_{p}}\pa_{z_{q}}\theta_{\tau_{\eps}}(0;H)/\theta_{\tau_{\eps}}(0;H)$,
where $\tau_{\eps}:=\tau_{\M_{\eps}}.$ By differentiating the definition
of the theta function with half-integer theta characteristic $H$,
\[
\theta_{\tau}(z,H)=\sum_{s\in S}\exp(Q_{\tau}(s)+s\cdot(z+(\pi\i/2)M)),
\]
 where $Q_{\tau}(s)=\frac{1}{4}\sum_{i,j}\tau_{ij}s_{i}s_{j}$, we
get

\begin{equation}
\frac{\pa_{z_{p}}\pa_{z_{q}}\theta_{\tau_{\eps}}(0;H)}{\theta_{\tau_{\eps}}(0;H)}=\frac{\sum_{s\in\setS}s_{k}s_{l}\exp(Q_{\tau_{\eps}}(s)+(\pi\i/2)(s\cdot M))}{\sum_{s\in\setS}\exp(Q_{\tau_{\eps}}(s)+(\pi\i/2)(s\cdot M))},\label{eq: theta ratio}
\end{equation}
where $M$ is as in Definition \ref{def:markchar}. We now apply Lemma
\ref{lem: limit_period} to obtain 
\begin{equation}
\exp(Q_{\tau_{\eps}}(s))=\eps^{d(s)}\exp\left(Q_{\tau_{0}+O(\eps)}(s)\right),\label{eq: expA}
\end{equation}
where 
\begin{equation}
d(s)=2\left(\sum_{k=2g+1}^{2g+n-1}\frac{s_{k}^{2}}{8}\right)+\sum_{2g+1\leq i<j\leq2g+n-1}\frac{s_{i}s_{j}}{4}=\sum_{k=2g+1}^{2g+n-1}\frac{s_{k}^{2}}{8}+\frac{1}{8}\left(\sum_{p=2g+1}^{2g+n-1}s_{k}\right)^{2}.\label{eq: d(s)}
\end{equation}
Recalling that $n$ is even, the minimal value of $d(s)$ is $d_{0}=\frac{n-1}{8}+\frac{1}{8}=\frac{n}{8},$
attained if and only if $s\in S_{0}.$We can write 
\[
\frac{\pa_{z_{p}}\pa_{z_{q}}\theta_{\tau_{\eps}}(0;H)}{\theta_{\tau_{\eps}}(0;H)}=\frac{\sum_{s\in\setS}\eps^{d(s)-d_{0}}s_{k}s_{l}\exp(Q_{\tau_{0}+O(\eps)}(s)+(\pi\i/2)(s\cdot M))}{\sum_{s\in\setS}\eps^{d(s)-d_{0}}\exp(Q_{\tau_{0}+O(\eps)}(s)+(\pi\i/2)(s\cdot M))}
\]
If we take the limit $\eps\to0$ in the numerator and the denominator,
the terms with $d(s)\neq d_{0}$ all tend to zero, leading to (\eqref{eq: RHS_limit}).
The exchange of the (infinite) summation and the passage to the limit
$\eps\to0$ is justified as in the proof of Lemma 4.9 in \cite{BIVW}.
\end{proof}
\begin{rem}
The expression $\sum_{s\in\setS}\exp(Q_{\tau_{\eps}}(s)+(\pi\i/2)(s\cdot M)))$
is equal to $\theta_{\tau_{\eps}}(0,H).$ Hence, if $\sum_{s\in\setSn}\exp(Q_{\tau_{0}}(s)+(\pi\i/2)(s\cdot M)))\neq0,$
then also $\theta_{\tau_{\eps}}(0,H)\neq0$ for $\eps$ small enough,
since the former is the leading term of the asymptotics of the latter.
By the results of \cite{Hejhal}, this implies that $\Lambda_{\M_{\eps},\SpStr_{\eps}}$
exists and is unique, i.e., the spin line bundle $\SpStr_{\eps}$
does not have non-trivial holomorphic sections. In remark \ref{rem: uniqueness}
below, we will sketch an alternative, self-contained proof of the
latter statement.
\end{rem}

We now consider the limits of linear combinations 
\begin{equation}
\hd_{\M_{\eps},Q}^{[\eta]}(\cdot)=\bar{\eta}\Lambda_{\M_{\eps},\SpStr_{\eps}}(\cdot,Q)+\eta\Lambda_{\M_{\eps},\SpStr_{\eps}}(\cdot,Q^{\star})\label{eq: lincomb}
\end{equation}
where $\eta\in\C,$ $Q\in\M$ and $Q^{\star}=\inv(Q)\in\Mtilde.$
This is a half-differential with respect to the first argument, but
to define it properly, we need to be careful about the choices of
the coordinate charts at $Q$ and $Q^{\star}.$ We consider a \emph{simply-connected}
symmetric domain $\mathcal{U}=\inv(\mathcal{U})$ that intersects
$\Ann_{\eps_{1}}(v_{1})$ and no other pinching annuli, and contains
$Q$ and $Q^{\star}.$ We let $z:\mathcal{U}\to\C$ be a local coordinate
such that $z(\inv(\cdot))\equiv\overline{z(\cdot)},$ and we define
$\hd_{\M_{\eps},Q,z}^{[\eta]}(\cdot)$ by (\ref{eq: lincomb}), using
the coordinate $z$ for the second argument of the Szegő kernels.
In other words, given $Q$, $\hd_{\M_{\eps}}^{[\eta]}(Q,\cdot)$ is
the unique meromorphic half-differential on $\M_{\eps}$ with simple
poles as $Q$ and $Q^{\star}$, such that its elements with respect
to the chart $z$ satisfy 
\begin{align}
\frac{\hd_{\M_{\eps}Q,z}^{[\eta]}(P)}{\sqrt{dz(P)}}=\frac{\bar{\eta}}{z(P)-z(Q)}+O(1); & \quad P\to Q,\label{eq: F_M_eps_expansion}\\
\frac{\hd_{\M_{\eps},Q,z}^{[\eta]}(P)}{\sqrt{dz(P)}}=\frac{\eta}{z(P)-z(Q^{\star})}+O(1), & \quad P\to Q^{\star}.\label{eq: F_M_eps_expansion_2}
\end{align}
The uniqueness follows from the fact that for $\eps_{i}$ small enough,
the spin line bundle corresponding to $\SpStr_{\eps}$ has no holomorphic
section. 
\begin{lem}
\label{lem: limit_LHS} Assuming that $\sum_{s\in\setSn}\exp(Q_{\tau_{0}}(s)+(\pi\i/2)(s\cdot M)))\neq0$,
as $\eps\to0,$ the linear combination $\hd_{\M_{\eps}Q,z}^{[\eta]}(\cdot)$
converges uniformly on $\M\setminus\{v_{1},\dots,v_{n},Q\}$ to (the
unique) holomorphic half-differential $\hd_{\M,Q,z}^{[\eta]}(\cdot)$
corresponding to the spin structure $\SpStr$ satisfying, for the
local coordinate $z$ as above and any local coordinate $w$, 
\begin{equation}
\frac{\hd_{\M,Q,z}^{[\eta]}(P)}{\sqrt{dz(P)}}=\frac{\bar{\eta}}{z(P)-z(Q)}+O(1),\quad P\to Q;\label{eq: spinor_pole}
\end{equation}
\begin{equation}
\frac{\hd_{\M,Q,z}^{[\eta]}(P)}{\sqrt{dw(P)}}=\frac{c_{i}e^{\i\pi/4}}{\sqrt{w(P)-w(v_{i})}}+O(1),\quad P\to v_{i};\label{eq: spinor_v_i}
\end{equation}
for some coefficients $c_{i}\in\R.$ 
\end{lem}

\begin{proof}
The locally uniform convergence to \emph{some} half-differential follows
readily from \ref{eq: bosonization} and Lemma \ref{limitlemma}.
It is also clear that windings of $\gamma_{01}$ and $\gamma_{10}$
with respect to $\hd_{\M_{\eps}Q,z}^{[\eta]}$ are preserved under
a passage to the limit, and since for all $\eps>0,$ $\hd_{\M_{\eps}}^{[\eta]}(Q,P)$
has a simple pole of residue $\bar{\eta}$ as $z\to w,$ hence so
will $\hd^{[\eta]}(w,z).$ It remains to prove (\ref{eq: spinor_v_i});
cf. \cite[Lemma 4.3]{BIVW}. Note that (\ref{eq: spinor_v_i}) is
equivalent to $(\hd_{\M,Q,z}^{[\eta]})^{2}$ having at most simple
pole at $v_{i},$ with residue of the form $\i c_{1}^{2}$ for $c_{1}$
real (which is a coordinate-independent condition). The fact that
the poles are at most simple follows from (\ref{eq: RHS_limit}),
hence, we only need to check the argument of the residues. 

We start with $i=1,$ and note that $\left(\overline{\inv^{*}\hd_{\M_{\eps}Q,z}^{[\eta]}}\right)=\hd_{\M_{\eps}Q,z}^{[\eta]}$,
since both sides are meromorphic half-differentials with the same
singularities. Working in the chart $z,$ this boils down to $\hd_{\M_{\eps}Q,z}^{[\eta]}/\sqrt{dz}=\overline{\hd_{\M_{\eps}Q,z}^{[\eta]}(\bar z)}/\sqrt{dz}.$
Hence, for $P$ such that $z\in\R,$ we have $(\hd_{\M_{\eps},Q,z}^{[\eta]}(P))^{2}/dz(P)\in\R_{\geq0},$
and thus $\int_{C^{(1)}}(\hd_{\M_{\eps},Q,z}^{[\eta]})^{2}\in\R_{\geq0},$
where we denote $C^{(i)}=\{z_{i}^{-1}(|\eps_{i}|^{\frac{1}{2}}e^{i\theta}):\theta\in[0,2\pi)\}.$
Therefore, the same is true for the integral over a circle of small
fixed radius around $v_{1},$ and we can pass to a limit to conclude
that the same is true for $(\hd_{\M,Q,z}^{[\eta]})^{2},$and thus
(\ref{eq: spinor_v_i}) for $i=1$ is proven.

To prove the result for $i=2,\dots,n$, we consider a symmetric tubular
neighborhood of $B_{2g+i-1}$, a symmetric \emph{doubly connected}
region $\mathcal{U}_{i}=\inv(\mathcal{U}_{i})$ on $\M_{\eps}$ that
crosses $\Ann_{\eps_{1}}(v_{1})$ and $\Ann_{\eps_{i}}(v_{i})$ and
intersects no other pinching annuli. Let $w_{i}:\mathcal{U}_{i}\to\C$
be a coordinate that maps $\mathcal{U}_{i}$, say, to a round annulus,
satisfying $w_{i}(\inv(\cdot))=\overline{w_{i}(\cdot)},$ and such
that $\M\cap\mathcal{U}_{i}$ is mapped to the upper half-plane. The
``equators'' $\{z_{1}^{-1}(|\eps_{1}|^{\frac{1}{2}}e^{\i\theta}):\theta\in[0,2\pi)\}$
and $\{z_{i}^{-1}(|\eps_{i}|^{\frac{1}{2}}e^{\i\theta}):\theta\in[0,2\pi)\}$
of the pinching regions are mapped to segments $I_{1},I_{i}\subset\R$
respectively. By the correspondence between half-integer theta characteristics
and spin structures, see \cite[Lemma 3.14]{BIVW}, we see that with
our choice of $M_{i},$ we must have $(2\pi)^{-1}\wind_{\SpStr_{\eps}}(B_{2g+n-1})=1\mod2.$
Since $\hd_{\M,Q,z}^{[\eta]}$ is a section of the spin line bundle
corresponding to the spin structure $\SpStr_{\eps},$ we can compute
the winding, working in the coordinate $w_{i},$ as the difference
of the rotation numbers of $\gamma'(t)$ and $\left(\overline{\hd_{\M,Q,z}^{[\eta]}/\sqrt{dw_{i}}}\right)^{2}(\gamma(t))$
where $\gamma$ is any parametrization of the path $w_{i}(B_{2g+n-1}).$
Since the rotation number of $\gamma'(t)$ is $2\pi,$ we conclude
that the rotation number of $\left(\overline{\hd_{\M,Q,z}^{[\eta]}/\sqrt{dw_{i}}}\right)^{2}(\gamma(t))$
along $B_{i}$ is zero modulo $4\pi.$ By symmetry, this means that
the rotation number of $\left(\overline{\hd_{\M,Q,z}^{[\eta]}/\sqrt{dw_{i}}}\right)^{2}(\gamma(t))$
along the ``upper half'' of $w_{i}(B_{2g+n-1})$ is zero modulo
$2\pi.$ But since $\left(\hd_{\M,Q,z}^{[\eta]}\right)^{2}/dw_{i}\in\R_{\geq0}$
on $I_{1},$ this implies that also $\left(\hd_{\M,Q,z}^{[\eta]}\right)^{2}/dw_{i}\in\R_{\geq0}$
on $I_{i}.$ Hence, $\int_{C^{(i)}}(\hd_{\M_{\eps},Q,z}^{[\eta]})^{2}\in\R_{\geq0},$
and we conclude the proof as above. 
\end{proof}
\begin{rem}
By inspecting the above proof, it is easy to see that setting $M_{i}=1$
for some $i$ will lead to (\ref{eq: spinor_v_i}) with $c_{i}\in\i\R$
instead of $c_{i}\in\R$. In terms of the underlying Ising model correlations,
this corresponds to inserting a disorder instead of spin at the corresponding
$v_{i}$ (at least as long a parity constraints allow for a construction
of the corresponding correlation function). 
\end{rem}

\begin{rem}
\label{rem: uniqueness}Assume that $\omega$ is a holomorphic section
of the spin line bundle $\SpStr_{\eps},$ and consider, for $\eta\in\C,$
$\omega^{[\eta]}:=\bar{\eta}\omega+\eta\overline{\inv^{*}\omega}.$
Arguing as above, we see that $\int_{C^{(i)}}\left(\omega^{[\eta]}\right)^{2}\in\R_{\geq0}$
for all $i=1,\dots,n.$ On the other hand, since $\left(\omega^{[\eta]}\right)^{2}$
has no poles, we have $\sum_{i=1}^{n}\int_{C^{(i)}}\left(\omega^{[\eta]}\right)^{2}=0.$
Therefore, $(\omega^{[\eta]})^{2}/dz_{i}$ vanishes identically on
every $C^{(i)},$ and hence by uniqueness of the analytic continuation,
$(\omega^{[\eta]})^{2}$ vanishes identically everywhere, and since
this is true for every $\eta,$ so does $\omega$. 
\end{rem}

\section{Explicit formulae on a torus.}

\label{sec: formulae_on_torus}In this section, we will specialize
the objects given in the previous section to the case $\M=\T_{\omega_{1},\omega_{2}}=\C/\{m\omega_{1}+n\omega_{2},m,n\in\Z\}$,
which will give explicit formulae for the observables $\ffs{a,z}{\vv}{(p,q)}$
and $\ffs{a,z}{\vv}{\star,(p,q)},$ as well as of the spin correlations,
on a continuous torus $\T.$ We make the following observations:
\begin{itemize}
\item The unique normalized Abelian differential of the first kind on $\T_{\omega_{1},\omega_{2}}$
is $u_{1}(z)=\frac{\pi\i}{\omega_{1}}\,dz,$ and the period matrix
consist of the single entry $\tau=\pi\i\frac{\omega_{2}}{\omega_{1}}.$ 
\item Let $\nu=1,2,3,4$ correspond to $(M_{1},N_{1})=(1,1),(0,1),(0,0),(1,0)$
respectively, and denote 
\[
\theta_{\nu}(z)\coloneqq(-1)^{M_{1}N_{1}}\theta_{\tau_{\T_{\omega_{1},\omega_{2}}}}\left(\frac{\pi\i z}{\om 1};H\right).
\]
Then, $\theta_{\nu}(z)$ is a Jacobi theta function, specifically
$\theta_{\nu}(z)\coloneqq\theta_{\nu}(\frac{2\pi z}{\om 1}|\frac{\tau}{\pi\i})$
in the notation of \cite[Chapter 20]{DLMF}.
\item The normalized differential of the third kind can be written as 
\begin{equation}
\omega_{\T_{\omega_{1},\omega_{2}},u,v}(z)=d_{z}\log\left(\frac{\theta_{1}(z-u)}{\theta_{1}(z-v)}\right)=\left(\frac{\theta_{1}'(z-u)}{\theta_{1}(z-u)}-\frac{\theta_{1}'(z-v)}{\theta_{1}(z-v)}\right)dz.\label{eq: torus_thirdkind}
\end{equation}
Indeed, since $\theta_{1}$ has a simple zero at the origin and no
other zeros, the right-hand side has simple poles of residues $\pm1$
at $u,v$ and no other poles. The periodicity of $\theta_{1}$ implies
that the period of the right-hand side is an integer multiple of $2\pi\i,$
but it depends continuously on $u,v$ and vanishes at $u=v.$ 
\item Based on (\ref{eq: torus_thirdkind}), we can now calculate $(\tau_{0})_{i+2,j+2}$
for $2\leq i<j\leq n.$ Note that those entries depend on a choice
of the pinching coordinates, which we take to be $z_{i}=R(z-v_{i}),$
with $R>0$ large enough that the discs $\{|z_{i}|\leq1\}=\{|z-v_{i}|\leq R^{-1}\}$
do not overlap. Recall \ref{eq: u_i_all_2} and Remark \ref{rem: symmetries}.
We then can write 
\begin{multline*}
\regint{\int_{v_{1}}^{v_{j}}\omega_{\T_{\omega_{1},\omega_{2}},v_{i},v_{1}}(z)}\\
=\int_{v_{1}}^{v_{j}}d_{z}\log(\theta_{1}(z-v_{i}))-\int_{v_{1}+R^{-1}}^{v_{j}}d_{z}\log(\theta_{1}(z-v_{1}))-\int_{v_{1}}^{v_{1}+R^{-1}}d_{z}\log\left(\frac{\theta_{1}(z-v_{1})}{z-v_{1}}\right)\\
=\log\left(\frac{\theta_{1}(v_{j}-v_{i})}{\theta_{1}(v_{1}-v_{i})}\right)+\log\left(\frac{\theta_{1}(R^{-1})}{\theta_{1}(v_{j}-v_{1})}\right)+\log\theta_{1}'(0)-\log\left(\frac{\theta_{1}(R^{-1})}{R^{-1}}\right)\\
=\log\left(\frac{\theta_{1}(v_{j}-v_{i})\theta_{1}'(0)}{\theta_{1}(v_{1}-v_{i})\theta_{1}(v_{j}-v_{1})}\right)-\log R,
\end{multline*}
so that $\tau_{i+2},_{j+2}=\log\left|\frac{\theta_{1}(v_{j}-v_{i})\theta_{1}'(0)}{\theta_{1}(v_{1}-v_{i})\theta_{1}(v_{j}-v_{1})}\right|-\log R.$ 
\item The canonical differential of the second kind can be written as
\begin{equation}
B(z,w)\,dzdw=d_{z}d_{w}\log\theta_{1}(z-w)=\frac{\theta_{1}'(z-w)^{2}-\theta_{1}(z-w)\theta''_{1}(z-w)}{(\theta_{1}(z-w))^{2}}\,dzdw,\label{eq: torus_secondkind}
\end{equation}
for example, by taking the limit $u\to v$ in (\ref{eq: torus_thirdkind}).
\end{itemize}
Before turning to deriving explicit formulae for the Ising correlations,
let us set up some notation. Put $\setSp=\{s'\in\{\pm1\}^{n},\sum_{j=1}^{n}s'_{j}=0\}$,
$v_{ij}=v_{i}-v_{j},$ and $\vavg=\frac{\sum_{j=1}^{n}s'_{j}v_{j}}{2}.$
Further, denote 
\[
\ProdTh:=\prod_{1\leq i<j\leq n}\left|\frac{\theta_{1}(v_{ij})}{\theta'_{1}(0)}\right|^{\frac{s'_{i}s'_{j}}{2}};
\]
\begin{equation}
D_{v_{1},\dots,v_{n}}^{\nu}:=\sum_{s'\in\setSp}\left|\theta_{\nu}\left(\vavg\right)\right|^{2}\ProdTh\label{eq: defD}
\end{equation}

\begin{multline}
A_{\vv}^{\nu}(a,z)=\sum_{s'\in\setSp}\left(\pa^{2}\left|\theta_{\nu}\right|^{2}\right)\left(\vavg\right)\ProdTh+\\
+\frac{1}{2}\sum_{s'\in\setSp}\sum_{j=1}^{n}s_{j}'\left(\frac{\theta_{1}'(z-v_{j})}{\theta_{1}(z-v_{j})}+\frac{\theta_{1}'(a-v_{j})}{\theta_{1}(a-v_{j})}\right)\left(\pa\left|\theta_{\nu}\right|^{2}\right)\left(\vavg\right)\ProdTh+\\
+\frac{1}{4}\sum_{s'\in\setSp}\sum_{i,j=1}^{n}\left(s'_{i}s'_{j}\frac{\theta_{1}'(z-v_{i})}{\theta_{1}(z-v_{i})}\frac{\theta_{1}'(a-v_{j})}{\theta_{1}(a-v_{j})}\right)\left|\theta_{\nu}\left(\vavg\right)\right|^{2}\ProdTh
\end{multline}

\begin{multline}
B_{\vv}^{\nu}(a,z)=-\sum_{s'\in\setSp}\left(\bar{\pa}\pa\left|\theta_{\nu}\right|^{2}\right)(\vavg)\ProdTh+\\
-\frac{1}{2}\sum_{s'\in\setSp}\sum_{j=1}^{n}s_{j}'\left(\frac{\theta_{1}'(z-v_{j})}{\theta_{1}(z-v_{j})}+\bar{\left(\frac{\theta_{1}'(a-v_{j})}{\theta_{1}(a-v_{j})}\right)}\right)\left(\pa\left|\theta_{\nu}\right|^{2}\right)\left(\vavg\right)\ProdTh+\\
-\frac{1}{4}\sum_{s'\in\setSp}\sum_{i,j=1}^{n}s'_{i}s'_{j}\left(\frac{\theta_{1}'(z-v_{i})}{\theta_{1}(z-v_{i})}\bar{\left(\frac{\theta_{1}'(a-v_{j})}{\theta_{1}(a-v_{j})}\right)}\right)\left|\theta_{\nu}\left(\vavg\right)\right|^{2}\ProdTh
\end{multline}

\[
\]

\begin{thm}
We have 

\begin{align}
\left(\ffs{a,z}{\vv}{(p,q)}\right)^{2} & =4B(a,z)+4\frac{A_{\vv}^{\nu}(a,z)}{D_{v_{1},\dots,v_{n}}^{\nu}},\label{eq: f_explicit}\\
\left(\ffs{a,z}{\vv}{(*,p,q)}\right)^{2} & =4\frac{B_{\vv}^{\nu}(a,z)}{D_{v_{1},\dots,v_{n}}^{\nu}},\label{eq: f_star_explicit}
\end{align}
\end{thm}

\begin{rem}
The formulae above only give us the squares of $\ffs{a,z}{\vv}{(p,q)}$ and $\ffs{a,z}{\vv}{(*,p,q)}$, respectively. For the former one, the branch of the square root (on $\hat{\T}_{[v_1,\dots,v_n]}$) is identified by the expansion as $z\to a$. (Note that we only need $\ffs{a,z}{\vv}{(p,q)}$ to compute the spin correlations, and we don't need to worry about the sign, as we will see in the proof of Lemma \ref{cor: spin_compute} below). For the energy correlations, we actually do need to know the sign of $\ffs{a,z}{\vv}{(*,p,q)}$, which can be fixed e.g. from the condition 
$$
\lim_{z\to v_1} e^{-\i\pi/4}(z-v_1)^{\frac12}\ffs{a,z}{\vv}{(p,q)}=\overline{\lim_{z\to v_1} e^{-\i\pi/4}(z-v_1)^{\frac12}\ffs{a,z}{\vv}{\star(p,q)}},
$$
cf. the operator product expansions \cite[Theorem 6.2]{CHI2}.
\end{rem}

\begin{proof}
Working in the standard coordinate $z$ on the torus, Lemma \ref{lem: limit_LHS}
implies that we have that for each $\eta$ the locally uniform convergence
\[
\mathcal{F}_{(\T_{\omega_{1},\omega_{2}})_{\eps},a,z}^{[\eta]}(z)/\sqrt{dz}\to\ffs{a,z}{\vv}{[\eta],(p,q)}.
\]
 (Strictly speaking, we cannot use $z$ as a local coordinate as in
Lemma \ref{lem: limit_LHS} since it conjugates $\inv$ to a reflection
with respect to a circle rather than the real line, but with that
change the same argument applies). Hence, $\Lambda_{(\T_{\omega_{1},\omega_{2}})_{\eps},\SpStr_{\eps}}(z,a)/\sqrt{dz}$
and $\Lambda_{(\T_{\omega_{1},\omega_{2}})_{\eps},\SpStr_{\eps}}(z,a^{\star})/\sqrt{dz}$
converge to $\ffs{a,z}{\vv}{(p,q)}$ and $\ffs{a,z}{\vv}{\star,(p,q)}$
respectively. To prove (\ref{eq: f_explicit}), (\ref{eq: f_star_explicit}),
we need to compare the expressions given above with (\ref{eq: RHS_limit}).
Let's first redefine $s'_{i}\coloneqq s_{i+1}$, for $i=2,3,\cdots,n$,
and $s'_{1}\coloneqq-\sum_{i=2}^{n}s'_{i}$. Now, if $s\in\setS_0$,
then $s'_{i}\in\setSp.$

We proceed by writing the terms in the denominator in (\ref{eq: RHS_limit}),
using Lemma \ref{lem: limit_period}, Remark \ref{rem: symmetries},
and the explicit formulae for $(\tau_{0})_{ij}$ in the torus case
collected at the beginning of the present section:

\begin{multline}
\exp(Q_{\tau_{0}}(s)+(\pi\i/2)(s\cdot M))\\
=C\cdot\exp\left(\frac{1}{4}s_{1}^{2}\pi\i\tau+\frac{1}{4}s_{2}^{2}\bar{\pi\i\tau}+\frac{\pi\i s_{1}}{2\om 1}\sum_{j=2}^{n}s'_{j}v_{j1}+\frac{\bar{\pi\i}s_{2}}{2\bar{\om 1}}\sum_{j=2}^{n}s'_{j}\bar{v_{j1}}+\frac{\pi\i}{2}s_{1}p+\frac{\pi\i}{2}s_{2}p\right)\prod_{2\leq i<j\leq n}e^{\tau_{ij}\frac{s'_{i}s'_{j}}{2}}\\
=C'\exp\left(\frac{1}{4}s_{1}^{2}\pi\i\tau+\frac{\pi\i s_{1}}{\om 1}\vavg+\frac{\pi\i}{2}s_{1}p\right)\overline{\exp\left(\frac{1}{4}s_{2}^{2}\pi\i\tau+\frac{\pi\i s_{2}}{\om 1}\vavg+\frac{\pi\i}{2}s_{2}p\right)}\\
\times\prod_{2\leq i<j\leq n}\left|\frac{\theta_{1}(v_{ij})\theta'_{1}(0)}{\theta_{1}(v_{1j})\theta_{1}(v_{1i})}\right|^{\frac{s'_{i}s'_{j}}{2}},\label{eq: control_denom}
\end{multline}
using $s'_{1}=-\sum_{i=2}^{n}s'_{i}$ in the second equality. Here
$C=\prod_{2\leq i\leq n}e^{\frac{\tau_{ii}\left(s_{i}'\right)^{2}}{4}}=\prod_{2\leq i\leq n}e^{\frac{\tau_{ii}}{4}}$
and $C'=(-1)^{pq}C\cdot R^{-\sum_{2\leq i<j\leq n}\frac{s_{i}'s_{j}'}{4}},$
where the first factor comes from the identity $e^{-\frac{\pi\i}{2}(2m_{2}+q)p}=(-)^{pq}e^{-\frac{\pi\i}{2}(2m_{2}+q)p}$
for $m_{2}\in\Z.$ Note that $C'$ does not depend on $s$ since $\sum_{2\leq i<j\leq n}\frac{s_{i}'s_{j}'}{4}=\frac{1}{8}\left(\sum_{i=2}^{n}s_{i}'\right)^{2}-\frac{1}{8}\sum_{i=2}^{n}\left(s_{i}'\right)^{2}=\frac{2-n}{8}.$
The last product can be rewritten as

\[
\prod_{2\leq i<j\leq n}\left|\frac{\theta_{1}(v_{ij})\theta'_{1}(0)}{\theta_{1}(v_{1j})\theta_{1}(v_{1i})}\right|^{\frac{s'_{i}s'_{j}}{2}}=\prod_{i=2}^{n}\left|\frac{\theta_{1}(v_{1i})}{\theta'_{1}(0)}\right|^{\frac{1}{2}}\prod_{1\leq i<j\leq n}\left|\frac{\theta_{1}(v_{ij})}{\theta'_{1}(0)}\right|^{\frac{s'_{i}s'_{j}}{2}}
\]
where we have used $s'_{1}=-\sum_{i=2}^{n}s'_{i}$. Summing (\ref{eq: control_denom})
first over $s_{1},s_{2}$, this product can be taken out, and the
exponentials sum to $\theta_{\nu}\left(\vavg\right).$ Thus, we conclude
that 

\begin{equation}
\sum_{s\in\setSn}\exp(Q_{\tau_{0}}(s)+(\pi\i/2)(s\cdot M))=C'\prod_{i=2}^{n}\left|\frac{\theta_{1}(v_{1i})}{\theta'_{1}(0)}\right|^{\frac{1}{2}}D_{v_{1},\dots,v_{n}}^{\nu}.\label{eq: denom_contin}
\end{equation}

We continue by looking at the terms in the numerator of (\ref{eq: RHS_limit}),
for various $k,l$ . First, observe that $u_{2}\equiv0$ on $\T$.
Therefore, when computing $\ffs{a,z}{\vv}{(p,q)}$, we will have the
following type of terms:
\begin{enumerate}
\item The term with $k=l=1$ gives a contribution of 
\begin{multline*}
\left(\frac{\pi\i}{\om 1}\right)^{2}\sum_{s\in\setSn}s_{1}^{2}\exp(Q_{\tau_{0}}(s)+(\pi\i/2)(s\cdot M))dzda\\
=C'\prod_{i=2}^{n}\left|\frac{\theta_{1}(v_{1i})}{\theta'_{1}(0)}\right|^{\frac{1}{2}}\sum_{s'\in\setSp}\left(\pa^{2}\left|\theta_{\nu}\right|^{2}\right)\left(\vavg\right)\ProdTh\,dzda,
\end{multline*}
by a computation completely paralled to the above, using (\ref{eq: control_denom})
and the identity 
\[
\left(\frac{\pi\i}{\om 1}\right)^{2}\sum_{s_{1}}s_{1}^{2}\exp\left(\frac{1}{4}s_{1}^{2}\pi\i\tau+\frac{\pi\i s_{1}}{\om 1}\vavg+\frac{\pi\i}{2}s_{1}p\right)=\pa^{2}\theta_{\nu}(\vavg).
\]
 
\item The terms with $k=1,3\leq l=j+1\leq n+1$ together give a contribution
of 
\begin{multline*}
\left(\frac{\pi\i}{\om 1}\right)\sum_{j=2}^{n}\frac{1}{2}\omega_{v_{j},v_{1}}(z)\sum_{s\in\setSn}s_{1}s'_{j}\exp(Q_{\tau_{0}}(s)+(\pi\i/2)(s\cdot M))da\\
=\frac{1}{2}\left(\frac{\pi\i}{\om 1}\right)\sum_{s\in\setSn}s_{1}\left(\sum_{j=2}^{n}s_{j}'\omega_{v_{j},v_{1}}(z)\right)\exp(Q_{\tau_{0}}(s)+(\pi\i/2)(s\cdot M))\,da\\
=C'\frac{1}{2}\prod_{i=2}^{n}\left|\frac{\theta_{1}(v_{1i})}{\theta'_{1}(0)}\right|^{\frac{1}{2}}\sum_{s'\in\setSp}\left(\sum_{j=1}^{n}s_{j}'\frac{\theta_{1}'(z-v_{j})}{\theta_{1}(z-v_{j})}\right)\left(\pa\left|\theta_{\nu}\right|^{2}\right)\left(\vavg\right)\ProdTh\,dzda,
\end{multline*}
where we have used that 
\[
\sum_{j=2}^{n}s_{j}'\omega_{v_{j},v_{1}}(z)=\sum_{j=2}^{n}\left(s_{j}'\frac{\theta_{1}'(z-v_{j})}{\theta_{1}(z-v_{j})}\right)-\left(\sum_{j=2}^{n}s_{j}'\right)\frac{\theta_{1}'(z-v_{1})}{\theta_{1}(z-v_{1})}=\sum_{j=1}^{n}s_{j}'\frac{\theta_{1}'(z-v_{j})}{\theta_{1}(z-v_{j})},
\]
the expansion $(\ref{eq: control_denom}),$and the identity 
\[
\left(\frac{\pi\i}{\om 1}\right)\sum_{s_{1}}s_{1}\exp\left(\frac{1}{4}s_{1}^{2}\pi\i\tau+\frac{\pi\i s_{1}}{\om 1}\vavg+\frac{\pi\i}{2}s_{1}p\right)=\pa\theta_{\nu}(\vavg).
\]
 
\item $l=1,3\leq k=i+1\leq n+1$ gives a symmetric contribution 
\[
=\frac{C'}{2}\prod_{i=2}^{n}\left|\frac{\theta_{1}(v_{1i})}{\theta'_{1}(0)}\right|^{\frac{1}{2}}\sum_{s'\in\setSp}\left(\sum_{i=1}^{n}s_{i}'\frac{\theta_{1}'(a-v_{i})}{\theta_{1}(a-v_{i})}\right)\left(\pa\left|\theta_{\nu}\right|^{2}\right)\left(\vavg\right)\ProdTh\,dzda,
\]
\item Finally, the terms with $3\leq k=i+1\leq n+1$ and $3\leq l=j+1\leq n+1$
give a contribution 
\begin{multline*}
\frac{1}{4}\sum_{i,j=2}^{n}\omega_{v_{i},v_{1}}(z)\omega_{v_{j},v_{1}}(a)\sum_{s\in\setSn}s'_{i}s'_{j}\exp(Q_{\tau_{0}}(s)+(\pi\i/2)(s\cdot M))\\
=\frac{1}{4}\sum_{s\in\setSn}\sum_{i,j=2}^{n}s'_{i}s'_{j}\omega_{v_{i},v_{1}}(z)\omega_{v_{j},v_{1}}(a)\exp(Q_{\tau_{0}}(s)+(\pi\i/2)(s\cdot M))\\
=\frac{C'}{4}\prod_{i=2}^{n}\left|\frac{\theta_{1}(v_{1i})}{\theta'_{1}(0)}\right|^{\frac{1}{2}}\sum_{s'\in\setSp}\left(\sum_{i,j=1}^{n}s_{i}'s_{j}'\frac{\theta_{1}'(z-v_{i})}{\theta_{1}(z-v_{i})}\frac{\theta_{1}'(a-v_{i})}{\theta_{1}(a-v_{i})}\right)\left|\theta_{\nu}\right|^{2}\left(\vavg\right)\ProdTh\,dzda
\end{multline*}
\[
\]
\end{enumerate}
Putting all the contributions together and dividing by (\ref{eq: denom_contin})
gives (\ref{eq: f_explicit}). The derivation of (\ref{eq: f_star_explicit})
is similar, with the following changes: we now need to evaluate the
limit of $\Lambda_{\T_{\eps},\SpStr_{\eps}}(z,a^{\star})$ for $z\in\T_{\omega_{1},\omega_{2}}$
and $a^{\star}\in\T_{\omega_{1},\omega_{2}}^{\star}$. The Abelian
differentials on $\T_{\omega_{1},\omega_{2}}^{\star}$ are pull-backs
of those on $\T_{\omega_{1},\omega_{2}}$ under the anti-holomorphic
involution $\inv,$ followed by conjugation and with a minus sign, cf. Remark \ref{rem: symmetries}. We will now have $u_{1}(a^{\star})\equiv u_{2}(z)\equiv0,$
so, instead of the term $k=l=1$, we will have the term $k=1,l=2.$ We leave it to the reader to check that following the above calculation with these changes, one arrives at (\ref{eq: f_star_explicit}).
\end{proof}
\begin{cor}
\label{cor: spin_compute}
We have, explicitly, 
\begin{equation}
\corr{\svv\mu^{\sps}}{\T_{\omega_{1},\omega_{2}}}=C_{n}\left(D_{v_{1},\dots,v_{n}}^{\nu}\right)^{\frac{1}{2}}\quad\text{where }C_{n}=\frac{2^{-\frac{n}{4}}}{|\theta_{2}(0)|+|\theta_{3}(0)|+|\theta_{4}(0)|}.\label{eq: ccor_explicit}
\end{equation}
\end{cor}

\begin{proof}
We will first check that 
\[
\partial_{v_{1}}\log\left(D_{v_{1},\dots,v_{n}}^{\nu}\right)^{\frac{1}{2}}=\mathcal{A}_{[v_{1},\dots v_{n}]}^{(p,q)};
\]
we start by computing the latter quantity using (\ref{eq: def_A})
and (\ref{eq: f_explicit}). First we note that we can actually compute $\mathcal{A}_{[v_{1},\dots v_{n}]}^{(p,q)}$ using the expansion of $(\ffs{w,z}{\vv}{(p,q)})^2$ instead of $\ffs{w,z}{\vv}{(p,q)}$ via the following expansion

\[
(\ffs{w,z}{\vv}{(p,q)})^2=\frac{\i}{w-v_1}\left(\gfs{z}{\vv}{(p,q)}+O(w-v_1)\right)
\]
where $\gfs{z}{\vv}{(p,q)}$ has the following expansion

\[
\gfs{z}{\vv}{(p,q)}=\frac{-\i}{z-v_1}\left(1+4\mathcal{A}_{[v_{1},\dots v_{n}]}^{(p,q)}(z-v_1)+O((z-v_1)^2)\right)
\]
Thus, we can write $\mathcal{A}_{[v_{1},\dots v_{n}]}^{(p,q)}$ as follows

\begin{equation}
4\mathcal{A}_{[v_{1},\dots,v_{n}]}^{(p,q)}=\lim_{z\to v_{1}}(z-v_{1})^{-1}\left[(z-v_{1})\lim_{w\to v_{1}}(w-v_{1})(\ffs{w,z}{\vv}{(p,q)})^2-1\right].
\end{equation}
We now note that since $\theta_{1}$ is odd,
we have 
\begin{equation}
\frac{\theta_{1}'(w-v_{1})}{\theta_{1}(w-v_{1})}=\frac{1}{w-v_{1}}+O(w-v_{1}),\quad w\to v_{1}.\label{eq: exp_logder_theta}
\end{equation}
and this is the only expression in (\ref{eq: f_explicit}) that is
not regular as $w\to v_{1}.$ Therefore, we get 
\begin{multline*}
\lim_{w\rightarrow v_{1}}(\ffs{w,z}{\vv}{(p,q)})^2(w-v_{1})=\left(2\sum_{s'\in\setSp}s_{1}'\left(\pa\left|\theta_{\nu}\right|^{2}\right)\left(\vavg\right)\ProdTh\right.\\
\left.+\sum_{s'\in\setSp}\sum_{i=1}^{n}\left(s'_{i}s'_{1}\frac{\theta_{1}'(z-v_{i})}{\theta_{1}(z-v_{i})}\right)\left|\theta_{\nu}\left(\vavg\right)\right|^{2}\ProdTh\right)/\left(D_{v_{1},\dots,v_{n}}^{\nu}\right)
\end{multline*}
We now need to compute the first two terms of the asymptotics of this
expression as $z\to v_{1};$ in view of (\ref{eq: exp_logder_theta}),
the leading term $(z-v_{1})^{-1}$ will come from the $i=1$
term in the second sum, which will not give a contribution to the
sub-leading one. We arrive at 
\begin{multline*}
2\mathcal{A}_{[v_{1},\dots v_{n}]}^{(p,q)}=\left(\sum_{s'\in\setSp}s_{1}'\left(\pa\left|\theta_{\nu}\right|^{2}\right)\left(\vavg\right)\ProdTh\right.\\
+\left.\frac{1}{2}\sum_{s'\in\setSp}\sum_{i=2}^{n}s'_{i}s'_{1}\frac{\theta_{1}'(v_{1i})}{\theta_{1}(v_{1i})}\left|\theta_{\nu}\left(\vavg\right)\right|^{2}\ProdTh\right)/\left(D_{v_{1},\dots,v_{n}}^{\nu}\right)
\end{multline*}
On the other hand, one can readily see that differentiating $D_{v_{1},\dots,v_{n}}^{\nu}$
with respect to $v_{1}$ yields the same expression, with the first
term coming from differentiating $\left|\theta_{\nu}\left(\vavg\right)\right|^{2}=\left|\theta_{\nu}\left(\frac{1}{2}(s_{1}'v_{1}+\dots+s_{n}'v_{n})\right)\right|^{2}$
with respect to $v_{1},$ and the $i$-th term in the second sum coming
from differentiating the term $\left|\frac{\theta_{1}(v_{1}-v_{j})}{\theta'_{1}(0)}\right|^{\frac{s'_{i}s'_{j}}{2}}$
term in $\ProdTh.$

Now, the fact that $\left(D_{v_{1},\dots,v_{n}}^{\nu}\right)^{\frac{1}{2}}$
is symmetric with respect to permutation of $v_{i}$ and real implies
that $d\log\left(D_{v_{1},\dots,v_{n}}^{\nu}\right)^{\frac{1}{2}}$
is given by (\ref{eq: dR}). It remains to check that $C_{n}\sum_{\nu}\left(D_{v_{1},\dots,v_{n}}^{\nu}\right)^{\frac{1}{2}}$
satisfies the normalization condition (\ref{eq: merge_asymptotics}),
which we do recursively. First, we note that since $\left|\frac{\theta_{1}(v_{2}-v_{1})}{\theta_{1}'(0)}\right|\sim|v_{2}-v_{1}|,$
we have 
\[
\ProdTh\sim|v_{1}-v_{2}|^{\frac{s'_{1}s_{2}'}{2}}\prod_{\substack{1\leq i<j\leq n\\
\{i,j\}\neq\{1,2\}
}
}\left|\frac{\theta_{1}(v_{i}-v_{j})}{\theta_{1}'(0)}\right|^{\frac{s_{i}'s'_{j}}{2}}.
\]
If $s_{1}'=s_{2}',$ this will tend to zero; otherwise, we have $\left|\frac{\theta_{1}(v_{1}-v_{j})}{\theta_{1}'(0)}\right|^{\frac{s_{1}'s'_{j}}{2}}\left|\frac{\theta_{1}(v_{2}-v_{j})}{\theta_{1}'(0)}\right|^{\frac{s_{2}'s'_{j}}{2}}\to1$
as $v_{1}\to v_{2}$ for all $j.$ We conclude that, for $n\geq2,$
we have 
\[
D_{v_{1},\dots,v_{n}}^{\nu}\sim|v_{1}-v_{2}|^{-\frac{1}{2}}\sum_{s'\in\setSp:s_{1}'=-s_{2}'}|\theta_{\nu}(\tilde{v}_{s'_{3\to n}})|^{2}\Pi_{v_{3},\dots,v_{n}}^{s'_{3\to n}}=2|v_{1}-v_{2}|^{-\frac{1}{2}}D_{v_{3},\dots,v_{n}}^{\nu}.
\]
where $s'_{3\to n}=(s'_{3},\dots,s_{n}'),$ or in other words, 
\[
C_{n}\left(D_{v_{1},\dots,v_{n}}^{\nu}\right)^{\frac{1}{2}}\sim C_{n-2}|v_{1}-v_{2}|^{-\frac{1}{4}}\left(D_{v_{3},\dots,v_{n}}^{\nu}\right)^{\frac{1}{2}},\quad v_{1}\to v_{2}.
\]
In particular, for $n=2,$ we have 
\[
C_{2}^{2}D_{v_{1},\dots,v_{n}}^{\nu}\sim\frac{1}{(|\theta_{2}(0)|+|\theta_{3}(0)|+|\theta_{4}(0)|)^{2}}|v_{2}-v_{1}|^{-\frac{1}{2}}|\theta_{\nu}(0)|^{2}.
\]
Taking the square root and summing over $\nu$ (recall that $\theta_{1}(0)=0$),
we conclude inductively that $C_{n}\sum_{\nu}\left(D_{v_{1},\dots,v_{n}}^{\nu}\right)^{\frac{1}{2}}$
satisfies (\ref{eq: merge_asymptotics}), thus concluding the proof. 
\end{proof}


\begin{thebibliography}{10}


\bibitem{Basok} Basok, Mikhail. Dimers on Riemann surfaces and compactified
free field. Preprint arXiv:2309.14522. 

\bibitem{BIVW} Bayraktaroglu, B., Izyurov, K., Virtanen, T., and
Webb, K. Bosonization of primary fields for the critical Ising model
on multiply connected planar domains. Preprint arXiv:2312.02960. 

\bibitem{BPZ} Belavin, Alexander A., Polyakov, Alexander M. and Zamolodchikov,
Alexander B. "Infinite conformal symmetry in two-dimensional quantum
field theory." Nuclear Physics B 241.2 (1984): 333-380. 

\bibitem{CHI1} Chelkak, Dmitry; Hongler, Clément; Izyurov, Konstantin.
Conformal invariance of spin correlations in the planar Ising model.
Ann. of Math. (2) 181 (2015), no. 3, 1087--1138. 

\bibitem{CHI2} Chelkak, Dmitry; Hongler, Clément; Izyurov, Konstantin.
Correlations of primary fields in the critical Ising model. Preprint
arXiv:2103.10263. 

\bibitem{CCK} Chelkak, Dmitry, Cimasoni, David, Kassel, Adrien, Revisiting
the combinatorics of the 2D Ising model, Ann. Inst. Henri Poincaré
Comb. Phys. Interact. 4 (2017), 309--385.

\bibitem{Cimasoni} David Cimasoni. A generalized Kac-Ward formula.
Journal of Statistical Mechanics: Theory and Experiment, 2010(07):P07023,
2010. 1, 3

\bibitem{Cim2} David Cimasoni. The critical Ising model via Kac--Ward
matrices. Communications in Mathematical Physics, 316(1):99--126,
2012. 1, 

\bibitem{DFMS} Di Francesco, Philippe; Mathieu, Pierre; Sénéchal,
David Conformal field theory. Graduate Texts in Contemporary Physics.
Springer-Verlag, New York, 1997. xxii+890 pp. 

\bibitem{DiFSZ} Di Francesco, Ph, H. Saleur, and J. B. Zuber. "Critical
Ising correlation functions in the plane and on the torus." Nuclear
Physics B 290 (1987): 527-581. 

\bibitem{DLMF} NIST Digital Library of Mathematical Functions. https://dlmf.nist.gov/,
Release 1.2.0 of 2024-03-15. F. W. J. Olver, A. B. Olde Daalhuis,
D. W. Lozier, B. I. Schneider, R. F. Boisvert, C. W. Clark, B. R.
Miller, B. V. Saunders, H. S. Cohl, and M. A. McClain, eds.
 

\bibitem{Dubedat15} Dubédat, Julien. "Dimers and families of Cauchy-Riemann
operators I." Journal of the American Mathematical Society 28.4 (2015):
1063-1167. 

\bibitem{Dubedat} Dubédat, Julien. Exact bosonization of the Ising
model. Preprint arXiv:1112.4399. 

\bibitem{Hugoetc} Duminil-Copin, Hugo, and Lis, Marcin. "On the
double random current nesting field." Probability Theory and Related
Fields 175 (2019): 937-955. 

\bibitem{Fay} Fay, John D. Theta functions on Riemann surfaces. Lecture
Notes in Mathematics, Vol. 352. Springer-Verlag, Berlin-New York,
1973. iv+137 pp. 

\bibitem{FF} Arthur E. Ferdinand and Michael E. Fisher. Bounded and
inhomogeneous Ising models. I. Specific-heat anomaly of a finite lattice.
Phys. Rev., 185:832--846, Sep 1969

\bibitem{Fuks} Fuks, B. A. Theory of analytic functions of several
complex variables, American Mathematical Soc. (1963). 

\bibitem{Galizka - Pete} Galicza, Pál, and Gábor Pete. \textquotedbl Sparse
reconstruction in spin systems II: Ising and other factor of IID measures.\textquotedbl{}
arXiv preprint arXiv:2406.09232 (2024)

\bibitem{HS} Hawley, N.S. and Schiffer, M. Half-order differentials
on Riemann surfaces, Acta Math. 115 (1966) 199-236. 

\bibitem{Hejhal} Hejhal, Dennis A. Theta functions, kernel functions,
and Abelian integrals. Memoirs of the American Mathematical Society,
No. 129. American Mathematical Society, Providence, R.I., 1972. iii+112
pp. 

\bibitem{HonKytVik} Hongler, Clément, Kalle Kytölä, and Fredrik Viklund.
"Conformal field theory at the lattice level: discrete complex analysis
and Virasoro structure." Communications in Mathematical Physics 395.1
(2022): 1-58. 

\bibitem{IzyurovFree} Izyurov, K. (2015). Smirnov’s observable for
free boundary conditions, interfaces and crossing probabilities. Communications
in Mathematical Physics, 337(1), 225-252.

\bibitem{IKT} Izyurov, K., Kemppainen, A., and Tuisku, P. (2024).
Energy correlations in the critical Ising model on a torus. The Annals
of Applied Probability, 34(2), 1699-1729. 

\bibitem{Johnson} Johnson, Dennis. Spin structures and quadratic
forms on surfaces. Journal of the London Mathematical Society, 2(2),
(1980), 365-373. 

\bibitem{OK} Kaufman, B., and Onsager, L. (1949). Crystal statistics.
III. Short-range order in a binary Ising lattice. Physical Review,
76(8), 1244. 


\bibitem{MW} McCoy, Barry M., and Tai Tsun Wu. The two-dimensional
Ising model. Harvard University Press, 1973. 

\bibitem{Onsager} Onsager, L. (1944). Crystal statistics. I. A two-dimensional
model with an order-disorder transition. Physical Review, 65(3-4),
117. 

\bibitem{Palmer} Palmer, John. Planar Ising Correlations. Vol. 49.
Springer Science and Business Media, 2007. 

\bibitem{Salas} Jesús Salas. Exact finite-size-scaling corrections
to the critical two-dimensional Ising model on a torus. Journal of
Physics A: Mathematical and General, 34(7):1311--1331, feb 2001.

\bibitem{Smirnov} Smirnov, S. (2010). Conformal invariance in random
cluster models. I. Holmorphic fermions in the Ising model. Annals
of mathematics, 1435-1467

\bibitem{Yamada} Yamada, Akira.``Precise variational formulas for
abelian differentials." Kodai Math. J. 3 (1) 114 - 143, 1980. 

\bibitem{Yang} Yang, C. N. (1952). The spontaneous magnetization
of a two-dimensional Ising model. Physical Review, 85(5), 808.

\end{thebibliography}
\end{document}